\AtBeginDocument{%
  \paperwidth=\dimexpr
    1in + \oddsidemargin
    + \textwidth
    + 1in + \oddsidemargin
  \relax
  \paperheight=\dimexpr
    1in + \topmargin
    + \headheight + \headsep
    + \textheight
    + 1in + \topmargin
  \relax
  \usepackage[pass]{geometry}\relax
}

\RequirePackage{fix-cm}
\documentclass[smallcondensed]{svjour3}     
\smartqed 

\usepackage{graphicx}
\usepackage{mathptmx}
\usepackage{amssymb}
\usepackage{pdfcomment}
\usepackage{times}
\usepackage{tikz}
\usepackage{amsmath}
\usepackage{amsfonts}
\usepackage{mathtools}
\usepackage{verbatim}
\usepackage{bm}
\usepackage{bibentry}
\usepackage{mathrsfs}
\usepackage{listings}
\usepackage{caption,subcaption}
\usepackage{scrextend}
\usepackage{datetime}
\usepackage{hyperref}
\usepackage{float}
\usepackage{dsfont}

\DeclareSymbolFont{newfont}{OML}{cmm}{m}{it}
\DeclareMathSymbol{\Varrho}{3}{newfont}{37}

\usepackage{amsmath}
\usepackage{amssymb}
\usepackage{mathtools}
\usepackage{comment}
\usepackage{enumitem}
\usepackage{xcolor}
\usepackage{mathrsfs}
\usepackage{float}
\usepackage[scr=mma]{mathalpha}
\usepackage{bm}
\usepackage{amsbsy}
\usepackage{dsfont}

\begin{document}

\title{A Nonparametric, Mixed Effect, Maximum Likelihood Estimator for the Distribution of Random Parameters in Discrete-Time Abstract Parabolic Systems with Application to the Transdermal Transport of Alcohol 
}

\titlerunning{Nonparametric Probability Distribution Estimation Using the Mixed Effects Model}        

\author{Lernik Asserian$^1$ \and Susan E. Luczak$^{2,4}$ \and I.G. Rosen$^{3,4}$ 
}

\authorrunning{Lernik Asserian et al.} 

\institute{L. Asserian \\
              \email{lernik@stanford.edu} \\
              S. E. Luczak \\
              \email{luczak@usc.edu} 
              \\
              I.G. Rosen \\
              \email{grosen@usc.edu} \\
            $^1$ Department of Mathematics, Stanford University, Stanford, CA, USA
            \\
              $^2$ Department of Psychology, University of Southern California, Los Angeles, CA, USA \\
              $^3$ Modeling and Simulation Laboratory, Department of Mathematics, University of Southern California, Los Angeles, CA, USA
            \\
              $^4$ This study was funded in part by the National Institute on Alcohol Abuse and Alcoholism under Grant Numbers: R21AA017711 and R01AA026368, S.E.L. and I.G.R.
         \\ }


\maketitle


\begin{abstract}
The existence and consistency of a maximum likelihood estimator for the joint probability distribution of random parameters in discrete-time abstract parabolic systems are established by taking a nonparametric approach in the context of a mixed effects statistical model using a Prohorov metric framework on a set of feasible measures. A theoretical convergence result for a finite dimensional approximation scheme for computing the maximum likelihood estimator is also established and the efficacy of the approach is demonstrated by applying the scheme to the transdermal transport of alcohol modeled by a random parabolic PDE. Numerical studies included show that the maximum likelihood estimator is statistically consistent in that the convergence of the estimated distribution to the ``true" distribution is observed in an example involving simulated data. The algorithm developed is then  applied to two datasets collected using two different transdermal alcohol biosensors. Using the leave-one-out cross-validation method, we get an estimate for the distribution of the random parameters based on a training set. The input from a test drinking episode is then used to quantify the uncertainty propagated from the random parameters to the output of the model in the form of a $95\%$ error band surrounding the estimated output signal.

\keywords{Nonparametric estimation \and Mixed effects model \and Maximum likelihood estimation \and Prohorov metric \and Existence and consistency \and Random discrete time dynamical systems \and Random partial differential equations \and finite dimensional approximation and convergence \and Alcohol biosensor \and Transdermal alcohol concentration}
\end{abstract}


\section{Introduction}
\label{intro}

In clinical therapy, medical research, and law enforcement, the breathalyzer, developed by Borkenstein based on a redox reaction and Henry's law \cite{Labianca:1990}, is used to measure Breath Alcohol Concentration (BrAC), a surrogate for Blood Alcohol Concentration (BAC). Clinicians and researchers consider it to be reasonably accurate to substitute BrAC for BAC and in general, this continues to be the case across different environmental conditions and across different individuals \cite{Labianca:1990}. Nevertheless, collecting near-continuous BrAC samples accurately (i.e. obtaining a deep lung sample that is not contaminated by any existing alcohol remaining in the mouth) is challenging and often impractical in the field.

Most of the ethanol, the type of alcohol in alcoholic beverages, that enters the human body, is metabolized by the liver into other products that are then excreted. In addition, a  portion of ingested ethanol exits the body directly through exhalation and urination \cite{Sakai:2006} and approximately $1\%$ diffuses through the epidermal layer of the skin in the form of perspiration and sweat. The amount of alcohol excreted in this manner is quantified in the form of transdermal alcohol concentration (TAC). TAC has been shown to be largely positively correlated with BrAC and BAC \cite{Swift:2000}. However, the precise relationship between TAC and BrAC/BAC is complicated due to a number of confounding physiological, technological, and environmental factors including, but not limited to, the skin's epidermal layer thickness, porosity and tortuosity, the process of vasodilation as observed through blood pressure and flow rate, the underlying technology of the particular sensor being used, and ambient temperature and humidity.

Currently, there are a number of different biosensors based on a variety of analog principles that can measure TAC essentially continuously, passively, unobtrusively, and relatively accurately, and make it available for processing in real time. Some of these devices are already commercially available and more are on the way. Several of these biosensors, like the breathalyzer, rely on relatively standard fuel-cell technology (i.e. converting chemical energy into electricity through redox reactions) to effectively count the number of ethanol molecules that evaporate during perspiration from the epidermal layer of the skin in near-continuous time \cite{Marques:2009}. Figure (\ref{fig:devices}) shows two of these TAC measuring devices; the WrisTAS$^{TM}$7, developed by Giner, Inc. in Waltham, MA and the SCRAM CAM$\textsuperscript{\textregistered}$ (Secure Continuous Remote Alcohol Monitor), developed by Alcohol Monitoring Systems, Inc. (AMS) in Littleton, Colorado.

\begin{figure}[H]
\centering
\includegraphics[width=3.9cm ,height= 3.2cm]{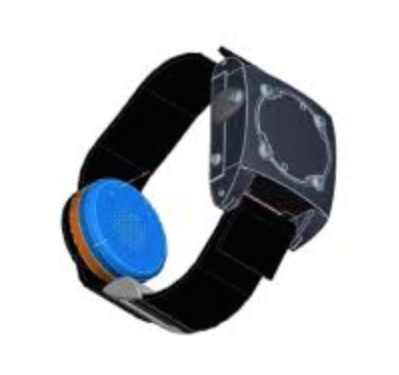}
\hspace{0.5 in}
\includegraphics[width=4.2cm ,height= 2.6cm]{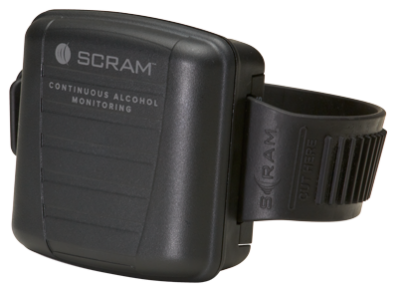}
\caption{WrisTAS$^{TM}$7 (left) and SCRAM CAM$\textsuperscript{\textregistered}$ (right) transdermal alcohol biosensors}
\label{fig:devices}
\end{figure}

Historically, researchers, clinicians, and the courts have always relied on BrAC or, when available, BAC. Consequently, in order to make TAC biosensors practical and accepted by the alcohol community, reliable and consistent means for converting TAC into equivalent BrAC or BAC must be developed. However, unfortunately, as indicated previously, a number of challenges must be dealt with before this can be done. In the past, our approach to developing a method for converting TAC into BrAC or BAC was based on deterministic methods for estimating parameters in distributed parameter systems such as those described in \cite{Banks:1997,Banks:1989}. Our earlier work along these lines has been reported in, for example, \cite{Dai:2016,Dumett:2008,Rosen:2014}. In these treatments, a forward model in the form of a one-dimensional diffusion equation based on Fick's law \cite{Smith:2004} with BrAC as the input and TAC as the output is first calibrated (i.e. fit) using BrAC and TAC data collected from the patient or research subject in the clinic or laboratory during what is known as a controlled alcohol challenge.  Then, after the same patient or research subject has worn the TAC sensor in the field for an extended period of time (e.g. days, weeks, or even months), the TAC data is downloaded, and the fit forward model is used to deconvolve the BrAC or BAC input from the observed TAC output. 

In order to eliminate the calibration process, we developed a population model-based approach wherein the parameters in the model were assumed to be random. Then, rather than fitting the values of the parameters themselves, their distributions were estimated based on BrAC and TAC data from a cohort of individuals (see, for example, \cite{Asserian:2021,Sirlanci:2019A,Sirlanci:2017,Sirlanci:2019B}). 

In all of our approaches to the TAC to BAC/BrAC conversion problem, the underlying model was taken to be based on the first principles physics based initial-boundary value problem for a parabolic partial differential equation.  This will also be the basic model to which we will direct our efforts in the present treatment. Transformed to be in terms of dimensionless variables, the model is given by
\begin{align}
    \frac{\partial x}{\partial t}(t,\eta) &= q_1 \frac{\partial^2 x}{\partial\eta^2}(t,\eta), \quad 0<\eta<1, \quad t>0, \label{eq.1}\\
    q_1\frac{\partial x}{\partial \eta}(t,0) &= x(t,0), \quad t>0, \label{eq.2}\\
    q_1\frac{\partial x}{\partial \eta}(t,1) &= q_2 u(t), \quad t>0, \label{eq.3}\\
    x(0,\eta) &= x_0, \quad 0<\eta<1, \label{eq.4}\\
    y(t) &= x(t,0), \quad t>0, \label{eq.5}
\end{align}
where $t$ and $\eta$ are the temporal and spatial variables, respectively, and $x(t,\eta)$ indicates the concentration of ethanol in the epidermal layer of the skin at time $t$ and depth $\eta$, where $\eta=0$ is at the skin surface and $\eta=1$ is at the boundary between the epidermal and dermal layers of the skin. The input to the system is $u(t)$, which is the BrAC/BAC at time $t$, and the output is $y(t)$, which is the TAC at time $t$. Equation (\ref{eq.1}) represents the transport of ethanol through the epidermal layer of the skin. The boundary conditions (\ref{eq.2}) and (\ref{eq.3}) represent respectively the evaporation of ethanol at the skin surface and the flux of ethanol across the boundary between the epidermal and dermal layers. It is assumed that there is no alcohol in the epidermal layer of the skin at time $t = 0$, so the initial condition (\ref{eq.4}) is $x_0(\eta)=0$, $0<\eta<1$. Finally, the output equation (\ref{eq.5}) represents the TAC level measured by the biosensor at the skin surface. 

The parameters in the system (\ref{eq.1})-(\ref{eq.5}) that will be assumed to be random are $q_1$ and $q_2$, which represent respectively the normalized diffusivity and the normalized flux gain at the boundary between the dermal and epidermal layers. The values or distributions of these parameters are assumed to depend on environmental conditions, the particular sensor being used, and the physiological characteristics of the individual wearing the sensor. The parameter vector is $\bm{q} = (q_1,q_2) \in Q$, where $Q$ is assumed to be a compact subset of $\mathbb{R}^+ \times \mathbb{R}^+$ with metric $d$.

In population modeling, we can statistically classify the methods as parametric or nonparametric. In the parametric approach, we assume that the general structure of the distribution is known a-priori but with unknown parameters. Then, for example, if we know that the distribution is normal with unknown mean and variance, the estimation problem is to estimate these two unknown parameters. On the other hand; in the nonparametric approach, the structure of the distribution is assumed to be unknown, and the problem is to estimate the distribution itself. In either of these paradigms, different statistical approaches to the estimation problem can be taken. For example, in \cite{Sirlanci:2019A,Sirlanci:2017,Sirlanci:2019B}, a  parametric least squares naive pooled data approach was used, while in \cite{Asserian:2021,Banks:2012,Banks:2018A}, the approach was nonparametric. In \cite{Hawekotte:2021}, a Bayesian framework was developed, and in the present treatment we consider a mixed effects (see, for example, \cite{Davidian:1995,Davidian:2003,Demidenko:2013}) maximum likelihood based statistical model. In the mixed effects model, it is assumed that observations are specific to a single individual plus a random error. The mixed effects model is a combination of the fixed-effects model, which describes the characteristics for an average individual in the population, and the random-effects model, which describes the inter-individual variability \cite{Lovern:2009}. An overview of these different statistical approaches in the context of pharmacokinetics can be found in \cite{Tatarinova:2013}. 

In addition to the work of our group on TAC to BAC/BrAC conversion cited above, other researchers have also been looking at this problem and have tried a number of different approaches.  For example, in \cite{Dougherty:2012,Dougherty:2015}, a more traditional approach based on standard linear regression techniques is developed and discussed. A number of ideas from the machine learning literature have also been considered. In \cite{Fairbairn:2021}, a scheme based on random forests is used to recover BrAC from TAC, and in our group, in \cite{Oszkinat:2021}, the authors develop a method using physics-informed neural networks, and in \cite{Oszkinat:2021A}, BrAC or BAC is estimated from observations of TAC using a Hidden Markov Model (HMM) based approach.      

An outline of the remainder of the paper is as follows. In Section \ref{sec:2}, we provide a summary of the Prohorov metric on the set of probability measures as it was used by Banks and his co-authors in \cite{Banks:2012}. In Section \ref{sec:3}, we define our mathematical model in the form of a random discrete-time dynamical system and we define the maximum likelihood estimator for the distribution of the random parameters. In the Section \ref{sec:4}, we establish the existence and consistency of the maximum likelihood estimator, while in Section \ref{sec:5}, we demonstrate the convergence of finite dimensional approximations for our estimator. In Section \ref{sec:6}, we summarize results for abstract parabolic systems, their finite dimensional approximation, and an associated convergence theory. In Section \ref{sec:7}, the application of our scheme to the transdermal transport of alcohol is presented and discussed.  This includes numerical studies for two examples, one involving simulated data and the other, actual data collected in the laboratory of one of the co-authors, Dr. Susan Luczak, in the Department of Psychology at University of Southern California (USC). For the simulated data example in Section \ref{sec:7.1}, we are able to observe the convergence of the estimated distribution of the random parameter vector $\bm{q}=(q_1,q_2)$ to the ``true" distribution as the number of drinking episodes increases, as the number of Dirac measure nodes increases, and as the level of discretization in the finite dimensional approximations increases.  We look at each case separately and in turn. In the actual data example discussed in Section \ref{sec:7.2}, we apply the leave-one-out cross-validation (LOOCV) method by first estimating the distribution of the parameter vector $\bm{q}$ using a training set, and then estimating the TAC output using the estimated distribution and the BrAC input of a testing episode.

\section{Prohorov Metric Framework}
\label{sec:2}

Banks and his co-authors developed a framework for estimation of the probability measure for random parameters in continuous-time dynamical systems based on the Prohorov metric \cite{Banks:2012}. Here, we summarize the Prohorov metric and its properties.

Let $Q$ be a Hausdorff metric space with metric $d$. Define $$C_b(Q) = \{f: Q \to \mathbb{R} \enskip | \enskip f \enskip \text{is bounded and continuous}\},$$ and given any probability measure $P \in \mathcal{P}(Q)$, where $\mathcal{P}(Q)$ denotes the set of all probability measures defined on $\Sigma_Q$, the Borel sigma algebra on $Q$, and some $\epsilon > 0$, an $\epsilon$-neighborhood of $P$ is defined by
    \begin{align*}
        B_\epsilon(P) = \bigg{\{} \tilde{P} \enskip \bigg{|} \enskip \Bigg{|} \int_Q f(\bm{q}) d\tilde{P}(\bm{q}) - \int_Q f(\bm{q}) dP(\bm{q}) \Bigg{|} < \epsilon, \enskip \text{for all} \enskip f \in C_b(Q) \bigg{\}}.
    \end{align*}
    
Let $E \in \Sigma_Q$, and define the $\epsilon$-neighborhood of $E$ by $$E^\epsilon = \{ \tilde{\bm{q}} \in Q \enskip | \enskip d(\tilde{\bm{q}},E) < \epsilon \} = \{ \tilde{\bm{q}} \in Q \enskip | \enskip \inf_{\bm{q} \in E} d(\bm{q},\tilde{\bm{q}}) < \epsilon \}.$$ 

Given two probability measures, $P$ and $\tilde{P}$ in $\mathcal{P}(Q)$, the Prohorov metric $\rho$ on $\mathcal{P}(Q) \times \mathcal{P}(Q)$ is defined such that $$\tilde{P} \in B_\epsilon(P) \iff \rho(P,\tilde{P}) < \epsilon,$$ where
    \begin{align*}
        \rho(P,\tilde{P}) = \inf \{ \epsilon > 0 \enskip | \enskip \tilde{P}(E) \leq P(E^\epsilon) + \epsilon \enskip \text{and} \enskip P(E) \leq \tilde{P}(E^\epsilon) + \epsilon, \enskip \text{for all} \enskip E \in \Sigma_Q \}.
    \end{align*}

It can be shown that $(\mathcal{P}(Q),\rho)$ is a metric space. Also, the Prohorov metric metrizes the weak convergence of measures, i.e. given a sequence of measures $P_M \in \mathcal{P}(Q)$, for all $M=1,2,\dots$, and $P \in \mathcal{P}(Q)$, $$P_M \xrightarrow{w^*} P \iff \rho(P_M,P) \to 0.$$ It is important to note that the weak$^*$ topology and the weak topology are equivalent on the space of probability measures.

For some $n_{\bm{q}}$, $Q \subseteq \mathbb{R}^{n_{\bm{q}}}$ and $P \in \mathcal{P}(Q)$, consider the random vector $X:Q \rightarrow \mathbb{R}^{n_{\bm{q}}}$ on the probability space $(Q,\Sigma_Q,P)$ given by $X(\bm{q}) = \bm{q}$ for $\bm{q} \in Q$.  The cumulative distribution function for $X$ is given by $F_X(q_1,\dots,q_n) = P(X \in \times_{\ell=1}^{n_{\bm{q}}}(-\infty,q_\ell])=P(\times_{\ell=1}^{n_{\bm{q}}}\{(-\infty,q_\ell]\cap Q\})$. In this case, it follows that if $\{P_M\}, P_0 \in (\mathcal{P}(Q)$, $\rho)$, then $\rho(P_M,P_0) \to 0$ if and only if $F_{X_M} \to F_{X_0}$ at all points of continuity of $F_{X_0}$. Consequently, Prohorov metric convergence and weak and weak$^*$ convergence in $\mathcal{P}(Q)$ are also referred to as convergence in distribution.

If $\bm{q}_1,\bm{q}_2 \in Q$, then $\rho(\delta_{\bm{q}_1},\delta_{\bm{q}_2}) = \min\{d(\bm{q}_1,\bm{q}_2),1\}$, where $\delta_{\bm{q}_j} \in D = \{\delta_{\bm{q}} \enskip | \enskip \bm{q} \in Q \},$ the space of Dirac measures on $Q$, where for all $E \in \Sigma_Q$,
\begin{align*}
    \delta_{\bm{q}}(E) =
    \begin{dcases}
    \enskip 1 & \enskip \text{if } \bm{q} \in E \\
    \enskip 0 & \enskip \text{if } \bm{q} \notin E
    \end{dcases}.
\end{align*}

The metric space $(\mathcal{P}(Q),\rho)$ is separable if and only if the metric space $(Q,d)$ is separable. The sequence $\{\bm{q}_j\}_{j=1}^\infty$ is Cauchy in $(Q,d)$ if and only if the sequence $\{\delta_{\bm{q}_j}\}_{j=1}^\infty$ is Cauchy in $(\mathcal{P}(Q),\rho)$. We also have $(Q,d)$ is complete if and only if $(\mathcal{P}(Q),\rho)$ is complete, and $(Q,d)$ is compact if and only if $(\mathcal{P}(Q),\rho)$ is compact. The details and proofs can be found in \cite{Banks:2012}.

Assume the metric space $(Q,d)$ is separable and let $Q_d = \{\bm{q}_j\}_{j=1}^{\infty}$ be a countable dense subset of $Q$. Define the dense (see \cite{Banks:2012}) subset of $\mathcal{P}(Q)$, $\tilde{\mathcal{P}}_d(Q)$, as 
\begin{align}
    \tilde{\mathcal{P}}_d(Q) = \{P \in \mathcal{P}(Q) \enskip | \enskip P=\sum_{j=1}^{M} p_j \delta_{\bm{q}_j}, \bm{q}_j \in Q_d, M \in \mathbb{N}, p_j \in [0,1] \cap \mathbb{Q}, \sum_{j=1}^{M} p_j =1\},
    \label{eq.Ptilde}
\end{align}
the collection of all convex combinations of Dirac measures on $Q$ with rational weights $p_j$ at nodes $\bm{q}_j \in Q_d$, and for each $M \in \mathbb{N}$ let
\begin{align*}
        \mathcal{P}_M(Q) = \{P \in \tilde{\mathcal{P}}_d(Q) \enskip | \enskip P=\sum_{j=1}^{M} p_j \delta_{\bm{q}_j}, \bm{q}_j \in \{\bm{q}_j\}_{j=1}^{M}\}.
\end{align*}


\section{The Mathematical Model}
\label{sec:3}

Consider the following discrete-time mathematical model for the $i^{th}$ subject at time-step $k$
\begin{align*}
    x_{k,i}(\bm{q}_i) &= g_{k-1}(x_{k-1,i}(\bm{q}_i),u_{k-1,i};\bm{q}_i), \enskip k = 1,\dots,n_i, \enskip i = 1,\dots,m,\\
    x_{0,i} &= \phi_{0,i}, \enskip i = 1,\dots,m,
\end{align*}
where $\bm{q}_i$ is the $i^{th}$ subject's parameter vector in $Q$, denoting the set of admissible parameters, $g_{k-1}: \mathcal{H} \times \mathbb{R}^{\nu} \times Q \to \mathcal{H}$, $\mathcal{H}$ is in general an infinite dimensional Hilbert space, and $u_{k-1,i} \in \mathbb{R}^{\nu}$ is the input. The output is given by
\begin{align*}
    y_{k,i}(\bm{q}_i) = h_k(x_{k,i}(\bm{q}_i), \phi_{0,i}, u_{k,i}; \bm{q}_i), \enskip k = 1,\dots,n_i, \enskip i = 1,\dots,m,
\end{align*}
where $h_k: \mathcal{H} \times \mathcal{H} \times \mathbb{R}^{\nu} \times Q \to \mathbb{R}$.

For the mixed effects model, we define
\begin{align}
    Y_{k,i} = y_{k,i}(\bm{q}_i) + e_{k,i}, \enskip k = 1,\dots,n_i, \enskip i = 1,\dots,m, \label{eq.MixedEffectsModel}
\end{align}
where for each $i=1,2,\dots,m$, $e_{k,i}$ are independent and identically distributed (i.i.d.) with mean $0$, variance $\sigma^2$, and  $e_{k,i}\sim \varphi$, $k=1,2,\dots,n_i$, where $\varphi$ is a density with respect to a sigma finite measure $\mu$ on $\mathbb{R}$ and assumed to be continuous on $\mathbb{R}$, and we assume that the random vectors $[e_{1,i},\dots, e_{n_i,i}]$ are independent with respect to $i$, $i = 1,2,\dots,m$. That is, we assume that the error is independent across individuals and conditionally independent within individuals (i.e. given $\bm{q}_i$). 
For each $i = 1,2,\dots,m$, let $\bm{Y}_i = [Y_{1,i},Y_{2,i},\dots,Y_{n_i,i}]^T$, $\bm{y}_i(\bm{q}_i) = [y_{1,i}(\bm{q}_i),y_{2,i}(\bm{q}_i),\dots,y_{n_i,i}(\bm{q}_i)]^T$, $\bm{e}_i = [e_{1,i},e_{2,i},\dots,e_{n_i,i}]^T$ and rewrite (\ref{eq.MixedEffectsModel}) as
\begin{align*}
    \bm{Y}_{i} = \bm{y}_{i}(\bm{q}_i) + \bm{e}_{i}, \enskip i = 1,\dots,m.
\end{align*}
Then, for $i=1,2,\dots,m$, $\bm{Y}_i$ are independent with $\bm{Y}_i \sim f_i(\cdot;\bm{q}_i):\mathbb{R}^{n_i} \rightarrow \mathbb{R}$, where
\begin{align}
        f_i(\bm{v};\bm{q}_i) = \prod_{k=1}^{n_i}\varphi(v_{k}-y_{k,i}(\bm{q}_i)), \enskip i=1,\dots,m,
\label{eq.f}
\end{align}
where $\bm{v} = [v_1,v_2,\dots,v_{n_i}]^T \in \mathbb{R}^{n_i}$.
Let $P \in \mathcal{P}(Q)$ denote a probability measure on the Borel sigma algebra on $Q$, $\Sigma_Q$, where $\mathcal{P}(Q)$ denotes the set of all probability measures defined on $\Sigma_Q$, and let $P_0 \in \mathcal{P}(Q)$ be the ``true" distribution of the random vector $\bm{q}_i$. The goal is to find an estimate of $P_0$. In order to generate an estimator for $P_0$, and establish theoretical results and computational tools, we use the nonparametric maximum likelihood (NPML) approach, introduced by Lindsay and Mallet in \cite{Lindsay:1983,Mallet:1986}, using the Prohorov metric-based framework on $\mathcal{P}(Q)$, introduced by Banks and his co-authors in \cite{Banks:2012}, and summarized in Section \ref{sec:2}. 

For $P \in \mathcal{P}(Q)$ and $i=1,2,\dots,m$, let
    \begin{align}
        \mathcal{L}_{i}(P;\bm{Y}_i) &= \int_Q f_i(\bm{Y}_i;\bm{q}_i) dP(\bm{q}_i) \nonumber\\
        &= \int_Q  \prod_{k=1}^{n_i} \varphi(Y_{k,i}-h_k(x_{k,i}(\bm{q}_i), \phi_{0,i}, u_{k,i}; \bm{q}_i)) dP(\bm{q}_i)
        \label{eq.Li}
    \end{align}
be the contribution of the $i^{th}$ subject to the likelihood function
\begin{align}
    \mathcal{L}_{\bm{n},m}(P;\bm{Y}) &= \prod_{i=1}^m \mathcal{L}_{i}(P;\bm{Y}_i) \nonumber \\
    &= \prod_{i=1}^m \int_Q f_i(\bm{Y}_i;\bm{q}_i) dP(\bm{q}_i) \nonumber\\
    &= \prod_{i=1}^m \int_Q  \prod_{k=1}^{n_i} \varphi(Y_{k,i}-h_k(x_{k,i}(\bm{q}_i), \phi_{0,i}, u_{k,i}; \bm{q}_i)) dP(\bm{q}_i),
    \label{eq.Lnm}
\end{align}
where $\bm{n}=\{n_i\}_{i=1}^m$ and $\bm{Y} = \{\bm{Y}_i\}_{i=1}^{m}$. The goal is to find $P$ that maximizes the likelihood function. Define the estimator
\begin{align}
    P_{\bm{n},m} &= \arg \max_{P \in \mathcal{P}(Q)} \mathcal{L}_{\bm{n},m}(P;\bm{Y}) \nonumber\\
    &= \arg \max_{P \in \mathcal{P}(Q)} \prod_{i=1}^m \int_Q f_i(\bm{Y}_i;\bm{q}_i)  dP(\bm{q}_i) \nonumber\\
    &= \arg \max_{P \in \mathcal{P}(Q)} \prod_{i=1}^m \int_Q \prod_{k=1}^{n_i} \varphi(Y_{k,i}-h_k(x_{k,i}(\bm{q}_i), \phi_{0,i}, u_{k,i}; \bm{q}_i)) dP(\bm{q}_i). \label{eq.Pnm}
\end{align}

Let $\mathscr{\hat{y}}_{k,i}$ be realizations of the random variables $Y_{k,i}$, and define
\begin{align}
    \hat{P}_{\bm{n},m} &= \arg \max_{P \in \mathcal{P}(Q)} \mathcal{L}_{\bm{n},m}(P;\pmb{\mathscr{\hat{y}}}) \nonumber\\
    &= \arg \max_{P \in \mathcal{P}(Q)} \prod_{i=1}^m \int_Q f_i(\pmb{\mathscr{\hat{y}}}_i;\bm{q}_i)  dP(\bm{q}_i) \nonumber\\
    &= \arg \max_{P \in \mathcal{P}(Q)} \prod_{i=1}^m \int_Q \prod_{k=1}^{n_i} \varphi(\mathscr{\hat{y}}_{k,i}-h_k(x_{k,i}(\bm{q}_i), \phi_{0,i}, u_{k,i}; \bm{q}_i)) dP(\bm{q}_i),
    \label{Phatnm}
\end{align}
where $\pmb{\mathscr{\hat{y}}} = \{{\pmb{\mathscr{\hat{y}}}}_i\}_{i=1}^{m}$, with ${\pmb{\mathscr{\hat{y}}}}_i = [\mathscr{\hat{y}}_{1,i},\mathscr{\hat{y}}_{2,i},\dots,\mathscr{\hat{y}}_{n_i,i}]^T$.

The results of Lindsay and Mallet in \cite{Lindsay:1983,Mallet:1986} states that the maximum likelihood estimator $\hat{P}_{\bm{n},m}$ can be found in the class of discrete distributions with at most $m$ support points, i.e. $\hat{P}_{\bm{n},m} \in \mathcal{P}_M(Q)$, where $M \leq m$.

We cannot exactly compute the maximum likelihood estimator, $\hat{P}_{\bm{n},m}$, since $y_{k,i}$ must be approximated numerically by $y^N_{k,i}$ using a Galerkin numerical scheme with $N$ denoting the level of discretization. Also, similarly, we define our approximating estimator over the set $\mathcal{P}_M(Q)$ where $M$ denotes the number of nodes, $\{\bm{q}_j\}_{j=1}^M$. As a result, the optimization is over a finite set of parameters, being the rational weights $\{p_j\}_{j=1}^M$. Thus, our approximating estimator is
\begin{align}
    \hat{P}_{\bm{n},m,M}^N &= \arg \max_{P \in \mathcal{P}_M(Q)} \mathcal{L}_{\bm{n},m}^N(P;\pmb{\mathscr{\hat{y}}}) \nonumber\\
    &= \arg \max_{P \in \mathcal{P}_M(Q)} \prod_{i=1}^m \int_Q f_i^N(\pmb{\mathscr{\hat{y}}}_i;\bm{q}_i) dP(\bm{q}_i) \nonumber\\
    &= \arg \max_{P \in \mathcal{P}_M(Q)} \prod_{i=1}^m \int_Q \prod_{k=1}^{n_i} \varphi(\mathscr{\hat{y}}_{k,i}-h_k(x^N_{k,i}(\bm{q}_i), \phi_{0,i}, u_{k,i}; \bm{q}_i)) dP(\bm{q}_i). \label{eq.PhatnmNM}
\end{align}


\section{Existence and Consistency of the Maximum Likelihood Estimator}
\label{sec:4}

In \cite{Lindsay:1983}, the existence and uniqueness of a maximum likelihood estimator of a mixing distribution using the geometry of mixture likelihoods was established. Similarly, in \cite{Mallet:1986}, the existence and uniqueness of the maximum likelihood estimator for the distribution of the parameters of a random coefficient regression model was established. Here we provide an existence argument based on the maximization of a continuous function over a compact set.

The following theorem establishes the existence of the estimator $\hat{P}_{\bm{n},m}$ in (\ref{eq.PhatnmNM}), obtained from the realizations $\{\mathscr{\hat{y}}_{k,i}\}, \enskip k=1,\dots,n_i, \enskip i=1,\dots,m$ of the random variables $\{Y_{k,i}\}, \enskip k=1,\dots,n_i, \enskip i=1,\dots,m$. This is sufficient for establishing the existence of the maximum likelihood estimator $P_{\bm{n},m}$ in (\ref{eq.Pnm}).

\begin{theorem}
    For $i=1,2,\dots,m$, let $\mathcal{L}_{i}$ be given by equation (\ref{eq.Li}), and let $\mathcal{L}_{\bm{n},m}(P;\pmb{\mathscr{\hat{y}}})$ be given by equation (\ref{eq.Lnm}) where $\pmb{\mathscr{\hat{y}}} = \{{\pmb{\mathscr{\hat{y}}}}_i\}_{i=1}^{m}$, with ${\pmb{\mathscr{\hat{y}}}}_i = [\mathscr{\hat{y}}_{1,i},\mathscr{\hat{y}}_{2,i},\dots,\mathscr{\hat{y}}_{n_i,i}]^T$.
    Assume that for each $\pmb{\mathscr{\hat{y}}}_i$, we have a continuous function $\mathcal{L}_{i}(.;\pmb{\mathscr{\hat{y}}}_i) : \mathcal{P}(Q) \to \mathbb{R}$, and also for each $P \in \mathcal{P}(Q)$ we have a measurable function $\mathcal{L}_{i}(P;.) : \mathbb{R}^{n_i} \to \mathbb{R}$. Then there exists a measurable function $\hat{P}_{\bm{n},m} : \prod_{i=1}^m \mathbb{R}^{n_i} \to \mathcal{P}(Q)$ such that $$\mathcal{L}_{\bm{n},m}(\hat{P}_{\bm{n},m};\pmb{\mathscr{\hat{y}}}) = \sup_{P \in \mathcal{P}(Q)} \mathcal{L}_{\bm{n},m}(P;\pmb{\mathscr{\hat{y}}}).$$
\end{theorem}

\begin{proof}
    The theorem can be proven in a similar way as in \cite{Banks:2012} with the difference that we are taking the sup (instead of the inf) of a continuous function over a compact set.
    \qed
\end{proof}

In order to establish the consistency of the maximum likelihood estimator $P_{\bm{n},m}$, we show that $\rho(P_{\bm{n},m},P_0)$ converges almost surely to zero.  We do this by applying a theorem by Kiefer and Wolfowitz in \cite{Kiefer:1956}, establishing that the nonparametric maximum likelihood approach is statistically consistent. In other words, as the number of subjects, $m$, gets larger, the estimator $P_{\bm{n},m}$ converges in probability to $P_0$, the ``true" distribution, in the sense of the Prohorov metric, or weakly, or in distribution. Here, we have set up our problem in a way that makes establishing the consistency a straightforward application of the consistency result in \cite{Kiefer:1956}. 

\begin{theorem}
For each $i=1,...,m$, assume that the map $\bm{q}_i \mapsto f_i(\bm{v};\bm{q}_i)$ from $Q$ into $\mathbb{R}$ is continuous for each $\bm{v} \in \mathbb{R}^{n_i}$, and $f_i(\bm{v};\bm{q}_i)$ is measurable in $\bm{v}$ for any $\bm{q}_i \in Q$, where $f_i$ is given by (\ref{eq.f}). Assume further that $P_0$ is identifiable; that is, for $P_1 \in \mathcal{P}(Q)$ with $P_1 \neq P_0$, we have 
\begin{align*}
  \prod_{i=1}^m \int_{0}^{\bm{z}_i} \int_Q  \prod_{k=1}^{n_i} \varphi(z_{k}-h_k(x_{k,i}(\bm{q}_i), \phi_{0,i}, u_{k,i}; \bm{q}_i)) dP_1(\bm{q}_i)d\mu^{n_i} & \\
  \neq \prod_{i=1}^m \int_{0}^{\bm{z}_i} \int_Q  \prod_{k=1}^{n_i} \varphi(z_{k}-h_k(x_{k,i}(\bm{q}_i), \phi_{0,i}, u_{k,i}; \bm{q}_i)) dP_0(\bm{q}_i)d\mu^{n_i},
\end{align*}
for at least one $\bm{z} = [\bm{z}_1^T,\dots,\bm{z}_m^T]^T \in \mathbb{R}^{\sum_{i=1}^m n_i}$, where for $i=1,2,\dots,m$, $\bm{z}_i = [z_1,\dots,z_{n_i}]^T \in \mathbb{R}^{n_i}$, and the technical integrability assumption holds; that is, for any $P \in \mathcal{P}(Q)$,
    \begin{align*}
        \lim_{\epsilon \downarrow 0} E_{P_0} \Bigg{[} \log \frac{\sup\limits_{\tilde{P} \in B_\epsilon(P)} \mathcal{L}_{i}(\tilde{P};\bm{Y}_i)}{\mathcal{L}_{i}(P_0;\bm{Y}_i)} \Bigg{]}^+ < \infty,
    \end{align*}
where $\mathcal{L}_i$ is given by equation (\ref{eq.Li}). Then, as $m \to \infty$, $\rho(P_{\bm{n},m},P_0) \to 0$ almost surely (i.e with probability 1) and therefore in probability as well.
\end{theorem}

\begin{proof}
    The assumptions we have made in the previous section and in the statement of the theorem are sufficient to argue that assumptions 1-5 in \cite{Kiefer:1956} are satisfied. The conclusion of the consistency result in \cite{Kiefer:1956} is that the cumulative distribution functions, $F_{\bm{n},m}$, corresponding to $P_{\bm{n},m}$ converge almost surely to the cumulative distribution function $F_0$ corresponding to $P_0$ at every point of continuity of $F_0$. It follows that $\rho(P_{\bm{n},m},P_0) \to 0$ almost surely (i.e with probability 1); thus, in probability as well, as $m \to \infty$, and the theorem is proven.
    \qed
\end{proof}


\section{Convergence of the Finite Dimensional Approximations}
\label{sec:5}

We want to establish the convergence of the finite dimensional maximum likelihood estimators to the maximum likelihood estimator corresponding to the infinite dimensional model. As mentioned earlier, we cannot actually compute $\hat{P}_{\bm{n},m}$ in (\ref{Phatnm}) and consequently we approximate it by $\hat{P}_{\bm{n},m,M}^N$ in (\ref{eq.PhatnmNM}). Consider the following assumptions,
\begin{itemize}[leftmargin=0.25in]
    \item[\textbf{A1.}] For all $\bm{n}$, $m$, and $N$, the map $P \mapsto \mathcal{L}_{\bm{n},m}^N(P;\pmb{\mathscr{\hat{y}}})$ is a continuous map.
    
    \item[\textbf{A2.}] For any $\{P_M\}, P \in \mathcal{P}(Q)$ such that $\rho(P_M,P) \to 0$, we have $\mathcal{L}_{i}^N(P_M;\pmb{\mathscr{\hat{y}}}_i) \to \mathcal{L}_{i}(P;\pmb{\mathscr{\hat{y}}}_i)$ as $N,M \to \infty$ for $i=1,\dots,m$.
    
    \item[\textbf{A3.}] For all $P \in \mathcal{P}(Q)$ and $i=1,\dots,m$, $\mathcal{L}_{i}^N(P_M;\pmb{\mathscr{\hat{y}}}_i)$ and $\mathcal{L}_{i}(P;\pmb{\mathscr{\hat{y}}}_i)$ are uniformly bounded.
\end{itemize}

\begin{theorem}
    Under assumptions \emph{\textbf{A1-A3}}, there exists maximizers $\hat{P}_{\bm{n},m,M}^N$ given by equation (\ref{eq.PhatnmNM}).
    In addition, there exists a subsequence of $\hat{P}_{\bm{n},m,M}^N$ that converges to $\hat{P}_{\bm{n},m}$ given by equation (\ref{Phatnm}) as $M,N \to \infty$.
    \label{Thm:compconv}
\end{theorem}
\begin{proof}
For all $\bm{n}$, $m$, and $N$, by continuity of the map $P \mapsto \mathcal{L}_{\bm{n},m}^N(P;\pmb{\mathscr{\hat{y}}})$, per assumption \textbf{A1}, and compactness of $(\mathcal{P}(Q),\rho)$, we can conclude that $\hat{P}_{\bm{n},m,M}^N$ exists.

In \cite{Banks:2012}, it is shown that $\tilde{\mathcal{P}}_d(Q)$ given by equation (\ref{eq.Ptilde}) is a dense subset of $\mathcal{P}(Q)$. Thus, for $M=1,2,\dots$, construct a sequence of probability measures $P_M \in \mathcal{P}_M(Q) \subset \tilde{\mathcal{P}}_d(Q) \subset \mathcal{P}(Q)$, such that $\rho(P_M,P) \to 0$ in $\mathcal{P}(Q)$. Then by assumptions \textbf{A2} and \textbf{A3}, we have
\begin{align*}
   \bigg{|} \mathcal{L}_{\bm{n},m}^N(P_M;\pmb{\mathscr{\hat{y}}}) - \mathcal{L}_{\bm{n},m}(P;\pmb{\mathscr{\hat{y}}}) \bigg{|} &= \Bigg{|} \prod_{i=1}^m \mathcal{L}_{i}^N(P_M;\pmb{\mathscr{\hat{y}}}_i) - \prod_{i=1}^m \mathcal{L}_{i}(P;\pmb{\mathscr{\hat{y}}}_i) \Bigg{|} \to 0.
\end{align*}
Consequently, $\mathcal{L}_{\bm{n},m}^N(P_M;\pmb{\mathscr{\hat{y}}}) \to \mathcal{L}_{\bm{n},m}(P;\pmb{\mathscr{\hat{y}}})$ as $N,M \to \infty$.\\
In addition, by definition, for each $\bm{n}$, $m$, and $N$, and for all $P_M \in \mathcal{P}_M(Q)$, we have
\begin{align}
    \mathcal{L}_{\bm{n},m}^N(\hat{P}_{\bm{n},m,M}^N; \pmb{\mathscr{\hat{y}}}) \leq \mathcal{L}_{\bm{n},m}^N(P_M;\pmb{\mathscr{\hat{y}}}).
    \label{eq.inequality}
\end{align}
In addition, by compactness of $\mathcal{P}(Q)$, there exists a subsequence of $\hat{P}_{\bm{n},m,M}^N$ that converges to $\hat{P}_{\bm{n},m}$ as $M,N \to \infty$.
Thus, by taking the limit in (\ref{eq.inequality}) as $M,N \to \infty$, for all $P \in \mathcal{P}(Q)$, we find that
\begin{align*}
    \mathcal{L}_{\bm{n},m}(\hat{P}_{\bm{n},m}; \pmb{\mathscr{\hat{y}}}) \leq \mathcal{L}_{\bm{n},m}(P;\pmb{\mathscr{\hat{y}}});
\end{align*}
thus, $\hat{P}_{\bm{n},m} = \arg \max\limits_{P \in \mathcal{P}(Q)} \mathcal{L}_{\bm{n},m}(P;\pmb{\mathscr{\hat{y}}})$ as given in equation (\ref{Phatnm}).
\qed
\end{proof}

In practice, to achieve a desired level of accuracy, $M$ and $N$ are fixed sufficiently large. We choose a sufficiently large value for $N$, how large that needs to be, of course, depends on the particular numerical discretization scheme chosen. The most common choice would be using a Galerkin-based method to the approximate ${y}_{k,i}$ by ${y}_{k,i}^N$. We also choose a sufficiently large value for $M$, the number of nodes, $\{\bm{q}_j\}_{j=1}^{M}$. Therefore, the optimization problem is reduced to a standard constrained estimation problem over Euclidean $M$-space, in which we determine the values of the weights $p_j$ at each node $\bm{q}_j$ with the constraints that they all be non-negative and sum to one. By equation (\ref{eq.PhatnmNM}). It follows that
\begin{align*}
    \hat{P}_{\bm{n},m,M}^N &= \arg \max_{P \in \mathcal{P}_M(Q)} \mathcal{L}_{\bm{n},m}^N(P;\pmb{\mathscr{\hat{y}}}) \\
    &= \arg \max_{P \in \mathcal{P}_M(Q)} \prod_{i=1}^m \int_Q \prod_{k=1}^{n_i} \varphi(\mathscr{\hat{y}}_{k,i}-y^N_{k,i}(\bm{q}_i)) dP(\bm{q}_i)\\
    &= \arg \max_{P \in \mathcal{P}_M(Q)} \prod_{i=1}^m \int_Q \prod_{k=1}^{n_i} \varphi(\mathscr{\hat{y}}_{k,i}-h_k(x_{k,i}^N (\bm{q}_i), \phi_{0,i}^N, u_{k,i}; \bm{q}_i)) dP(\bm{q}_i)\\
    &= \arg \max_{\tilde{\bm{p}} \in \widetilde{\mathbb{R}^M}} \prod_{i=1}^m \sum_{j=1}^M \prod_{k=1}^{n_i} \varphi(\mathscr{\hat{y}}_{k,i}-h_k(x_{k,i}^N, \phi_{0,i}^N, u_{k,i}; \bm{q}_j)) p_j\\
    &= \arg \max_{\tilde{\bm{p}} \in \widetilde{\mathbb{R}^M}} \prod_{i=1}^m \sum_{j=1}^M p_j \prod_{k=1}^{n_i} \varphi(\mathscr{\hat{y}}_{k,i}-h_k(x_{k,i}^N, \phi_{0,i}^N, u_{k,i}; \bm{q}_j)),
\end{align*}
where $\tilde{\bm{p}} = (p_1,\dots,p_M) \in \widetilde{\mathbb{R}^M} = \{\tilde{\bm{p}} \enskip | \enskip p_j \in \mathbb{R}^+ , \sum_{j=1}^M p_j = 1\}$.

We note that computing $\hat{P}_{\bm{n},m,M}^N$ involves high order products of very small numbers which not unexpectedly can cause numerical underflow. In order to mitigate this, we maximize the log-likelihood function instead and rewrite it in a form that lends itself to the use of the MATLAB optimization routine \textit{logsumexp} as follows
\begin{align}
\label{eq.logsumexp}
    \hat{P}_{\bm{n},m,M}^N &= \arg \max_{\tilde{\bm{p}} \in \widetilde{\mathbb{R}^M}} \log\bigg{(}\prod_{i=1}^m \sum_{j=1}^M p_j \prod_{k=1}^{n_i} \varphi(\mathscr{\hat{y}}_{k,i}-h_k(x_{k,i}^N, \phi_{0,i}^N, u_{k,i}; \bm{q}_j))\bigg{)} \\
    &= \arg \max_{\tilde{\bm{p}} \in \widetilde{\mathbb{R}^M}} \sum_{i=1}^m \log\bigg{(}\sum_{j=1}^M p_j \prod_{k=1}^{n_i} \varphi(\mathscr{\hat{y}}_{k,i}-h_k(x_{k,i}^N, \phi_{0,i}^N, u_{k,i}; \bm{q}_j))\bigg{)} \nonumber\\
    &= \arg \max_{\tilde{\bm{p}} \in \widetilde{\mathbb{R}^M}} \sum_{i=1}^m \log\bigg{(}\sum_{j=1}^M \exp\bigg{(}\log\big{(}p_j \prod_{k=1}^{n_i} \varphi(\mathscr{\hat{y}}_{k,i}-h_k(x_{k,i}^N, \phi_{0,i}^N, u_{k,i}; \bm{q}_j))\big{)}\bigg{)}\bigg{)} \nonumber\\
    &= \arg \max_{\tilde{\bm{p}} \in \widetilde{\mathbb{R}^M}} \sum_{i=1}^m \log\bigg{(}\sum_{j=1}^M \exp\bigg{(}\log (p_j) + \sum_{k=1}^{n_i} \log\big{(} \varphi(\mathscr{\hat{y}}_{k,i}-h_k(x_{k,i}^N, \phi_{0,i}^N, u_{k,i}; \bm{q}_j))\big{)}\bigg{)}\bigg{)}. \nonumber
\end{align}


\section{Abstract Parabolic Systems}
\label{sec:6}
In order to apply our estimation theory to equations (\ref{eq.1})-(\ref{eq.5}), our model for the transdermal transport of ethanol given in Section \ref{intro}, we reformulate it as an abstract parabolic system. We briefly describe what an abstract parabolic system is, its properties, and its finite dimensional approximation, and then we show how assumptions \textbf{A1}-\textbf{A3} are satisfied for such a system. 

Let $H$ and $V$ be Hilbert spaces with $V$ densely and continuously embedded in $H$.  Pivoting on $H$, it follows that $H$ is therefore densely and continuously embedded in the dual of $V$, $V^*$. This is known as a Gelfand triple and is generally written as $V \hookrightarrow H \hookrightarrow V^*$ \cite{Tanabe:1979}. Then an abstract parabolic system is a dynamical system of the following form 
\begin{align}
    <\dot{x},\psi>_{V^*,V} + a(\bm{q};x,\psi) &= <\bm{B}(\bm{q})u,\psi>_{V^*,V}, \quad \psi \in V, \nonumber\\
    x(0) &= x_0, \label{eq.APS}\\
    y(t) &= \bm{C}(\bm{q}) x(t) \nonumber,
\end{align}
where $<\cdot,\cdot>_{V^*,V}$ denotes the duality pairing between $V^*$ and $V$, $Q$ is as defined in Section \ref{intro}, and for each $\bm{q} \in Q$, $a(\bm{q};.,.) : V \times V \to \mathbb{C}$ is a sesquilinear form satisfying the following three assumptions,
\begin{itemize}[leftmargin=0.25in]
    \item[\textbf{B1.}] \textbf{(Boundedness)} There exists a constant $\alpha_0$ such that for all $\psi_1,\psi_2 \in V$, we have
    \begin{align*}
        |a(\bm{q};\psi_1,\psi_2)| \leq \alpha_0 \parallel \psi_1 \parallel_V \parallel \psi_2 \parallel_V.
    \end{align*}
    \item[\textbf{B2.}] \textbf{(Coercivity)} There exists $\lambda_0 \in \mathbb{R}$ and $\mu_0>0$ such that for all $\psi \in V$, we have
    \begin{align*}
        a(\bm{q};\psi,\psi) + \lambda_0 |\psi|_H^2 \geq \mu_0 \parallel \psi \parallel_V^2.
    \end{align*}
    \item[\textbf{B3.}] \textbf{(Continuity)} For all $\psi_1,\psi_2 \in V$ and $\bm{q}, \tilde{\bm{q}} \in Q$, we have $$|a(\bm{q},\psi_1,\psi_2)-a(\tilde{\bm{q}},\psi_1,\psi_2)| \leq d(\bm{q},\tilde{\bm{q}}) \parallel \psi_1 \parallel_V \parallel \psi_2 \parallel_V.$$
\end{itemize}
In these assumptions, $\parallel.\parallel_V$ and $|.|_H$ denotes the norm on the spaces $V$ and $H$, respectively. Further, in (\ref{eq.APS}), $\bm{B}(\bm{q}): \mathbb{R}^{\nu} \to V^*$, and $\bm{C}(\bm{q}): V \to \mathbb{R}$ are bounded linear operators with initial conditions $x_0 \in H$, input $u \in L^2([0,T],\mathbb{R}^{\nu})$, and output $y \in L^2([0,T],\mathbb{R})$.

It can be shown that the system in (\ref{eq.APS}) has a unique solution in $$\big{\{} \psi \enskip | \enskip \psi \in L^2([0,T],V), \enskip \dot{\psi} \in L^2([0,T],V^*) \big{\}} \subset C([0,T],H)$$ using standard variational arguments (such as in \cite{Lions:1971}). However, we use a linear semigroup approach to convert the system in (\ref{eq.APS}) into a discrete-time state space model and then use arguments from linear semigroup theory \cite{Banks:1989,Pazy:1983} to argue convergence of finite dimensional Galerkin-based approximations and conclude that assumptions \textbf{A1}-\textbf{A3} are satisfied.

Assumptions \textbf{B1} and \textbf{B2} yield that the form $a(\bm{q};.,.)$ defines a bounded linear operator $\bm{A}(\bm{q}):V \to V^*$ given by
\begin{align*}
    <\bm{A}(\bm{q})\psi_1,\psi_2>_{V^*,V} = -a(\bm{q};\psi_1,\psi_2),
\end{align*}
where  for $\psi_1, \psi_2 \in V$. If we restrict the operator $\bm{A}(\bm{q})$ to the subspace $Dom(\bm{A}(\bm{q})) = \{\psi \in V \enskip | \enskip \bm{A}(\bm{q})\psi \in H\}$, it becomes the infinitesimal generator of a holomorphic or analytic semigroup, $\{e^{\bm{A}(\bm{q})t} \enskip | \enskip t \geq 0\}$, of bounded linear operators on $H$. The operator $\bm{A}(\bm{q})$ is referred to as being regularly dissipative \cite{Banks:1997,Banks:1989,Tanabe:1979}. Moreover, this semigroup can be extended and restricted to be a holomorphic semigroup on $V^*$ and $V$, respectively, as well \cite{Banks:1997,Tanabe:1979}.

The system in (\ref{eq.APS}) can now be written in state space form with time invariant operators $\bm{A}(\bm{q})$, $\bm{B}(\bm{q})$, and $\bm{C}(\bm{q})$, as 
\begin{align}
    \dot{x}(t) &= \bm{A}(\bm{q})x(t) + \bm{B}(\bm{q})u(t),\nonumber\\
    x(0) & = x_0, \label{eq.statespace}\\
    y(t) &= \bm{C}(\bm{q}) x(t) \nonumber.
\end{align}
The operator form of the variation of constants formula, then yields what is known as a mild solution of (\ref{eq.statespace}), and it is given by
\begin{align}
    x(t;\bm{q}) &= e^{\bm{A}(\bm{q}) t} x_0 + \int_0^t e^{\bm{A}(\bm{q})(t-s)} \bm{B}(\bm{q}) u(s) ds, \enskip t \geq 0, \label{eq.mildsol}\\
    y(t;\bm{q}) &= \bm{C}(\bm{q}) x(t;\bm{q}) \nonumber.
\end{align}

To obtain the corresponding discrete or sampled time form of the system in (\ref{eq.statespace}), let $\tau > 0$ be the length of the sampling interval, and consider strictly zero-order hold inputs of the form $u(t) = u_{k-1}, \enskip t \in [(k-1)\tau,k\tau), \enskip k = 1,2,\dots$. Then, let $ x_k = x(k\tau)$ and $y_k = y(k\tau), \enskip k = 1,2,\dots$. By applying (\ref{eq.mildsol}) on each sub-interval $[(k-1)\tau,k\tau]$, \enskip $k=1,2,\dots$, we obtain the discrete-time dynamical system given by
\begin{align}
    x_{k} &= \hat{\bm{A}}(\bm{q})x_{k-1} + \hat{\bm{B}}(\bm{q})u_{k-1}, \enskip k=1,2,\dots,\label{eq.x_k}\\
    y_{k} &= \hat{\bm{C}}(\bm{q}) x_k, \enskip k=1,2,\dots, \label{eq.y_k}
\end{align}
where $x_0 \in V$, $\hat{\bm{A}}(\bm{q}) = e^{\bm{A}(\bm{q}) \tau}$, $\hat{\bm{B}}(\bm{q}) = \int_0^\tau e^{\bm{A}(\bm{q}) s} \bm{B}(\bm{q}) ds$, and $\hat{\bm{C}}(\bm{q}) = \bm{C}(\bm{q})$.

Using a standard Galerkin approach \cite{Banks:1984}, we can approximate the discrete-time system given in (\ref{eq.x_k})-(\ref{eq.y_k}) by a sequence of approximating finite dimensional discrete-time systems in a sequence of finite dimensional subspaces, $V^N$, of $V$. In order to argue convergence, we will require the following additional assumption concerning the subspaces $V^N$,
\begin{itemize}[leftmargin=0.25in]
    \item[\textbf{C1.}] \textbf{(Approximation)} For every $x \in V$, there exists $x^N \in V^N$ such that $\parallel x-x^N \parallel_V \to 0$ as $N \to \infty$. 
\end{itemize}
We consider the sequence of approximating finite dimensional discrete-time systems by 
\begin{align*}
    x_{k}^N &= \hat{\bm{A}}^N(\bm{q})x_{k-1}^N + \hat{\bm{B}}^N(\bm{q})u_{k-1}, \enskip k=1,2,\dots,\\
    y_{k}^N &= \hat{\bm{C}}^N(\bm{q}) x_k^N, \enskip k=1,2,\dots,
\end{align*}
where $\hat{\bm{A}}^N(\bm{q}) = e^{\bm{A}^N(\bm{q}) \tau}$, $\hat{\bm{B}}^N(\bm{q}) = \int_0^\tau e^{\bm{A}^N(\bm{q}) s} \bm{B}^N(\bm{q}) ds$, and $\hat{\bm{C}}^N(\bm{q}) x_k = \hat{\bm{C}}(\bm{q})$, where for each $\bm{q} \in Q$, $\bm{A}^N(\bm{q})$ is the linear operator on $V^N$ obtained by restricting the form $a(\bm{q};.,.)$ to $V^N \times V^N$, i.e. for $\psi_1^N,\psi_2^N \in V^N$,
\begin{align*}
    <\bm{A}^N(\bm{q})\psi_1^N,\psi_2^N>_{V^*,V} = -a(\bm{q};\psi_1^N,\psi_2^N).
\end{align*}
And also, $\bm{B}^N(\bm{q}) = \pi^N \bm{B}(\bm{q})$, where in this definition, $\pi^N$ is the natural extension of the orthogonal projection operator $\pi^N: H \to V^N$ to $V^*$ from its dense subspace $H$. We also set $x_0^N = \pi^N x_0\in V^N$.

Under the assumptions \textbf{B1-B3} and \textbf{C1} using the Trotter-Kato approximation theorem from the theory of linear semigroups of operators \cite{Banks:1988,Pazy:1983}, we were able to conclude that $\lim_{N \rightarrow \infty}\parallel x_k^N - x_k\parallel_V = 0$ and $\lim_{N \rightarrow \infty} | y_k^N - y_k | = 0$ for each $x_0 \in V$, and uniformly in $\bm{q}$ for $\bm{q} \in Q$ and $k \in \{1,2,\dots,K\}$, for any fixed $K \in \mathbb{N}^+$.

We can now use the results described in the previous paragraphs to show that an abstract parabolic system satisfies assumptions \textbf{A1-A3} given in Section \ref{sec:5}. To show that the assumption \textbf{A1} is satisfied, we need to show that for all $\bm{n}$, $m$, and $N$, the map $P \mapsto \mathcal{L}_{\bm{n},m}^N(P;\pmb{\mathscr{\hat{y}}})$ is a continuous map. It suffices to show that for any fixed $\bm{n}$, $m$, and $N$, and for any sequence of probability measures $P_M$, such that $\rho(P_M,P) \to 0$ in $\mathcal{P}(Q)$, we have $\mathcal{L}_{\bm{n},m}^N(P_M;\pmb{\mathscr{\hat{y}}}) \to \mathcal{L}_{\bm{n},m}^N(P;\pmb{\mathscr{\hat{y}}})$ as $M \to \infty$. Towards this end, we see that
\begin{align*}
   \bigg{|} \mathcal{L}_{\bm{n},m}^N(P_M;\pmb{\mathscr{\hat{y}}}) - \mathcal{L}_{\bm{n},m}^N(P;\pmb{\mathscr{\hat{y}}}) \bigg{|} &= \Bigg{|} \prod_{i=1}^m \mathcal{L}_{i}^N(P_M;\pmb{\mathscr{\hat{y}}}_i) - \prod_{i=1}^m \mathcal{L}_{i}^N(P;\pmb{\mathscr{\hat{y}}}_i) \Bigg{|}\\
   &= \Bigg{|} \prod_{i=1}^m \int_Q \prod_{k=1}^{n_i} \varphi(\mathscr{\hat{y}}_{k,i}-y_{k,i}^N(\bm{q}_i)) dP_M(\bm{q}_i)\\
   & \enskip \enskip \enskip - \prod_{i=1}^m \int_Q \prod_{k=1}^{n_i} \varphi(\mathscr{\hat{y}}_{k,i}-y_{k,i}^N(\bm{q}_i)) dP(\bm{q}_i) \Bigg{|} \\
   & \to 0,
\end{align*}
by definition of the Prohorov metric. It follows that the assumption \textbf{A1} is satisfied.

Next, we show that the assumption \textbf{A2} is satisfied. We have that $\lim_{N \rightarrow \infty}\parallel x_{k,i}^N - x_{k,i}\parallel_V = 0$ and $\lim_{N \rightarrow \infty} | y_{k,i}^N - y_{k,i} | = 0$ for each $x_{0,i} \in V$, uniformly in $\bm{q}_i$ for $\bm{q}_i \in Q$, $k=1,\dots,n_i, \enskip i=1,\dots,m$. We want to show that for any sequence of probability measures $P_M$, such that $\rho(P_M,P) \to 0$ in $\mathcal{P}(Q)$, and for $i=1,\dots,m$, as $N, M \to \infty$, we have $\mathcal{L}_{i}^N(P_M;\pmb{\mathscr{\hat{y}}}_i) \to \mathcal{L}_{i}(P;\pmb{\mathscr{\hat{y}}}_i)$.

Recall that $\varphi$ is assumed to be continuous. Let $\epsilon > 0$, and choose $N_0$ such that for $N \geq N_0$, and for every $M$, $\Big{|} \int_Q \big{(} \prod_{k=1}^{n_i} \varphi(\mathscr{\hat{y}}_{k,i}-y_{k,i}^N(\bm{q}_i)) - \prod_{k=1}^{n_i} \varphi(\mathscr{\hat{y}}_{k,i}-y_{k,i}(\bm{q}_i)) \big{)} dP_M(\bm{q}_i) \Big{|} < \epsilon/2$. Then, we have
\begin{align*}
    & \Big{|} \mathcal{L}_{i}^N(P_M;\pmb{\mathscr{\hat{y}}}_i) - \mathcal{L}_{i}(P;\pmb{\mathscr{\hat{y}}}_i) \Big{|}\\
    & = \Big{|} \int_Q \prod_{k=1}^{n_i} \varphi(\mathscr{\hat{y}}_{k,i}-y_{k,i}^N(\bm{q}_i)) dP_M(\bm{q}_i) - \int_Q \prod_{k=1}^{n_i} \varphi(\mathscr{\hat{y}}_{k,i}-y_{k,i}(\bm{q}_i)) dP(\bm{q}_i) \Big{|} \\
    & \leq \Big{|} \int_Q \big{(} \prod_{k=1}^{n_i} \varphi(\mathscr{\hat{y}}_{k,i}-y_{k,i}^N(\bm{q}_i)) - \prod_{k=1}^{n_i} \varphi(\mathscr{\hat{y}}_{k,i}-y_{k,i}(\bm{q}_i)) \big{)} dP_M(\bm{q}_i) \Big{|} \\
    & \enskip + \Big{|} \int_Q \prod_{k=1}^{n_i} \varphi(\mathscr{\hat{y}}_{k,i}-y_{k,i}(\bm{q}_i)) dP_M(\bm{q}_i) - \int_Q \prod_{k=1}^{n_i} \varphi(\mathscr{\hat{y}}_{k,i}-y_{k,i}(\bm{q}_i)) dP(\bm{q}_i) \Big{|}\\
    & < \frac{\epsilon}{2} + \frac{\epsilon}{2} = \epsilon,
\end{align*}
where the second term is less than $\epsilon/2$ by definition of the Prohorov metric. Consequently, the assumption \textbf{A2} is satisfied.

Finally, we want to show that the assumption \textbf{A3} is satisfied. We want to show that for all $P \in \mathcal{P}(Q)$ and for $i=1,\dots,m$, $\mathcal{L}_{i}^N(P_M;\pmb{\mathscr{\hat{y}}}_i)$ and $\mathcal{L}_{i}(P;\pmb{\mathscr{\hat{y}}}_i)$ are uniformly bounded. Recall that the parameter space $Q$ is compact. Thus, for $\bm{q}_i \in Q$, and for each $N$, $y_{k,i}^N(\bm{q}_i)$ are uniformly bounded. Similarly, $y_{k,i}(\bm{q}_i)$ are also uniformly bounded and we also have that $|y_{k,i}^N(\bm{q}_i)-y_{k,i}(\bm{q}_i)| \to 0$ uniformly in $\bm{q}_i$ for $\bm{q}_i \in Q$. Therefore, we can conclude that the assumption \textbf{A3} is satisfied.


\section{Application to the Transdermal Transport of Alcohol}
\label{sec:7}

To apply the results established in Section \ref{sec:6} to the system (\ref{eq.1})-(\ref{eq.5}) in Section \ref{intro}, the system must first be written in weak form. Then, the parameter space $Q$, the Hilbert spaces $H$ and $V$, the sesquilinear form $a(\bm{q};.,.)$, and the operators $\bm{B}(\bm{q})$ and $\bm{C}(\bm{q})$ must all be identified. Also, the approximating subspaces, $V^N$, must be chosen, and finally assumptions \textbf{B1-B3} and \textbf{C1} must all be shown to be satisfied.

The parameter space $Q$ is assumed to be a compact subset of $\mathbb{R}^+ \times \mathbb{R}^+$ with any $p$-metric denoted by $d_Q$. Let $H=L^2(0,1)$ and $V=H^1(0,1)$ with their standard inner products and norms. It follows that $V^*=H^{-1}(0,1)$, and the three spaces $H$, $V$, and $V^*$ form a Gelfand triple. To rewrite the system (\ref{eq.1})-(\ref{eq.5}) in weak form, we multiply by a test function $\psi \in V$ and integrate by parts to obtain
\begin{align*}
    <\dot{x}(t),\psi>_{V^*,V} + \int_0^1 q_1 \frac{\partial x}{\partial \eta}(t,\eta) \psi'(\eta) d\eta + x(t,0) \psi(t,0) = q_2 u(t) \psi(1),
\end{align*}
where $<\cdot,\cdot>_{V^*,V}$ denotes the duality pairing between $V^*$ and $V$. Then for $\bm{q} \in Q$, $u \in \mathbb{R}$, and $\tilde{\psi}, \psi \in V$, we set
\begin{align*}
    a(\bm{q};\tilde{\psi},\psi) &= \int_0^1 q_1 \tilde{\psi}'(\eta) \psi'(\eta) d\eta + \tilde{\psi}(0) \psi(0),\\
    <\bm{B}(\bm{q})u,\psi>_{V^*,V} &= q_2 u\psi(1),\\
    \bm{C}(\bm{q})\psi&=\bm{C} \psi =\psi(0).
\end{align*}
We can establish that assumptions \textbf{B1-B3} are satisfied using arguments involving the Sobolev Embedding Theorem (see \cite{Adams:2003}). Also, the operators $\bm{B}(\bm{q})$ and $\bm{C}(\bm{q})$ are continuous in the uniform operator topology with respect to $\bm{q} \in Q$. It follows from Section \ref{sec:6} that
\begin{align*}
    &g_{k-1}(x_{k-1,i}(\bm{q}_i),u_{k-1,i};\bm{q}_i) = \hat{\bm{A}}(\bm{q}_i)x_{k-1,i}(\bm{q}_i) + \hat{\bm{B}}(\bm{q}_i)u_{k-1,i}, \enskip k=1,2,\dots,n_i, \enskip i=1,\dots,m,\\
    &h_k(x_{k,i}(\bm{q}_i), \phi_{0,i}, u_{k,i}; \bm{q}_i) = \hat{\bm{C}}(\bm{q}_i) x_{k,i}(\bm{q}_i), \enskip k = 1,\dots,n_i, \enskip i = 1,\dots,m,
\end{align*}
where $\hat{\bm{A}}(\bm{q}) = e^{\bm{A}(\bm{q}) \tau}$, $\hat{\bm{B}}(\bm{q}) = \int_0^\tau e^{\bm{A}(\bm{q}) s} \bm{B}(\bm{q}) ds$, and $\hat{\bm{C}}(\bm{q}) = \bm{C}(\bm{q})$ with $\tau>0$ the length of the sampling interval.

Let $V^N, \enskip N = 1, 2, \dots$, be the span of the standard linear splines defined with respect to the uniform mesh $\{0,1/N,2/N,\dots,(N-1)/N,1\}$ on $[0,1]$. Then, assumption \textbf{C1} is satisfied by standard arguments for spline functions (see, for example, \cite{Schultz:1973}). If for each $i=1,2,\dots,m$, we define $x_{k,i}^N$ and $y_{k,i}^N$ as in (\ref{eq.x_k})-(\ref{eq.y_k}), then by the arguments at the end of Section \ref{sec:6}, we conclude that assumptions \textbf{A1-A3} are satisfied.

In the following two subsections, \ref{sec:7.1} and \ref{sec:7.2}, we present the application of our scheme to the transdermal transport of alcohol in two examples, one involving simulated data, and the other using actual Human subject data collected in the Luczak laboratory at USC. For the simulated data, we want to show the convergence of the estimated distribution of the parameter vector $\bm{q}=(q_1,q_2)$ to the ``true" distribution as the number of drinking episodes increases, as the number of nodes increases, and as the level of discretization in the finite dimensional approximations increases. And, for the actual data, we apply the leave-one-out cross-validation (LOOCV) method by estimating the distribution of the parameter vector $\bm{q}$ using the training set, and then estimating the TAC output using the estimated distribution and the BrAC input of the test set.

\subsection{Example 1: Estimation Based on Simulation Data}
\label{sec:7.1}
In this example, we estimate the distribution of the parameter vector $\bm{q}=(q_1,q_2)$ in the system (\ref{eq.1})-(\ref{eq.5}) by first simulating TAC data in MATLAB with the assumption that the two parameters $q_1$ and $q_2$ are i.i.d. with a Beta distribution, $q_1,q_2 \sim Beta(2,5)$. Thus, their joint cumulative distribution function (cdf) is the product of their marginal $Beta(2,5)$ cdfs.

From equation (\ref{eq.MixedEffectsModel}), we have
\begin{align*}
    Y_{k,i} = y_{k,i}(\bm{q}_i) + e_{k,i}, \enskip k = 1,\dots,n_i, \enskip i = 1,\dots,m,
\end{align*}
where $m$ is the number of drinking episodes, and $y_{k,i}(\bm{q}_i)$ is the observed TAC for the $i^{th}$ drinking episode at time step $k$. We let $P_0$ be the product of the cdfs of two independent $Beta(2,5)$ distributions, and $e_{k,i} \sim N(0,10^{-6}), \enskip k = 1,\dots,n_i, \enskip i = 1,\dots,m$.

To approximate the PDE model for the TAC observations, we used the linear spline-based Galerkin approximation scheme described in Section \ref{sec:6} with $N$ equally spaced sub-intervals from $[0,1]$ (see \cite{Sirlanci:2019A,Sirlanci:2017,Sirlanci:2019B}). We want to compute $\hat{P}_{\bm{n},m,M}^N$ given by equation (\ref{eq.logsumexp}), where $\bm{q}_j = (q_{j_1},q_{j_2})$ is chosen as $M$ uniform meshgrid coordinates on $[0,1] \times [0,1]$. We make the assumption that there is no alcohol in the epidermal layer of the skin at time $t=0$, so we let $\phi_{0,i}^N = 0$. The constrained optimization problem over Euclidean $M$-space was solved using constrained optimization routine FMINCON from the Optimization Toolbox in MATLAB applied to the negative of the log-likelihood function.

In our earlier treatment \cite{Asserian:2021} in which the assumed observation was aggregated TAC, the appropriate underlying statistical model was the naive pooled model (i.e. the data point for each drinking episode at a certain time is an observation of the mean behavior plus a random error). When the nonlinear least squares-based constrained optimization problem was solved, the inherent ill-posedness of the inverse problem resulted in undesirable oscillations. To mitigate this behavior, we included an appropriately weighted regularization term in the performance index being minimized. One advantage of the mixed effects statistical model presented here is that regularization was not required.

In order to demonstrate the consistency of our estimator, we show that as $m$, the number of drinking episodes, increases, the estimated cdf of the parameter vector $\bm{q}=(q_1,q_2)$ approaches the ``true" cdf, the product of two $Beta(2,5)$ cdfs. In order to simulate realistic longitudinal TAC vectors representing data that might be collected by the TAC biosensor for an individual's drinking episode, we used BrAC data collected in the Luczak laboratory as the input to the model, and generated random samples of $q_1$ and $q_2$, i.i.d. $Beta(2,5)$ in MATLAB. Using the algorithm developed in the current paper, we estimated the distribution of the random parameter vector by solving the optimization problem for different cases based on the number of drinking episodes, $m \in \{1,3,7,9,16,42\}$, and observed the convergence of the estimated distribution to the ``true" distribution as $m$ increases.

To quantify this, let $D$ be the sum of the squared differences at each node between the estimated and the ``true" distribution, the product of two $Beta(2,5)$ cdfs. Let $p_j$ and $b_j$ be the weights at the node $\bm{q}_j$ of the estimated and ``true" distribution, respectively. Then,
\begin{align*}
    D = \sum_{j=1}^{M} (p_j - b_j)^2.
\end{align*}

We fixed the number of nodes, $M$, and the level of discretization, $N$, sufficiently large. We set $M=400$ and $N=128$. We estimated the distribution of the parameter vector for different cases based on different numbers of drinking episodes, $m \in \{1,3,7,9,16,42\}$, and calculated $D$ for each case. Our results are summarized in Table (\ref{tab:1}). We observed that as the number of drinking episodes, $m$, increases, the sum of the squared differences at each node between the estimated distribution and the ``true" distribution, $D$, decreases.
\begin{table}[H]
\centering
    \caption{Decrease in $D$, the sum of the squared differences at each node between the estimated and the ``true" distribution, with increasing $m$, the number of drinking episodes, for fixed values of the number of nodes $M=400$ and the level of discretization $N=128$.}
    \label{tab:1}
    \begin{tabular}{ll}
    \hline\noalign{\smallskip}
    {$m$} & {$D$} \\
    \noalign{\smallskip}\hline\noalign{\smallskip}
    1 & 39.0164 \\
    3 & 28.3091 \\
    7 & 8.3247 \\
    9 & 7.1750 \\
    16 & 3.5697 \\
    42 & 3.0337 \\
    \noalign{\smallskip}\hline
    \end{tabular}
\end{table}
In Figure (\ref{fig:3DSim}), we have plotted three different views of the estimated distribution and the ``true" distribution, (again, the product of two $Beta(2,5)$ cdfs), for the different cases where $m=7$, $m=16$, and $m=42$ drinking episodes in the top, middle, and bottom rows, respectively, with the number of nodes set to $M=400$ and the level of discretization to $N=128$. We observe that as $m$ increases, our estimated distribution gets ``closer" to the ``true" distribution, which agrees with the numerical results that are shown in Table (\ref{tab:1}).
\begin{figure}[H]
\centering
\includegraphics[width=10.55cm]{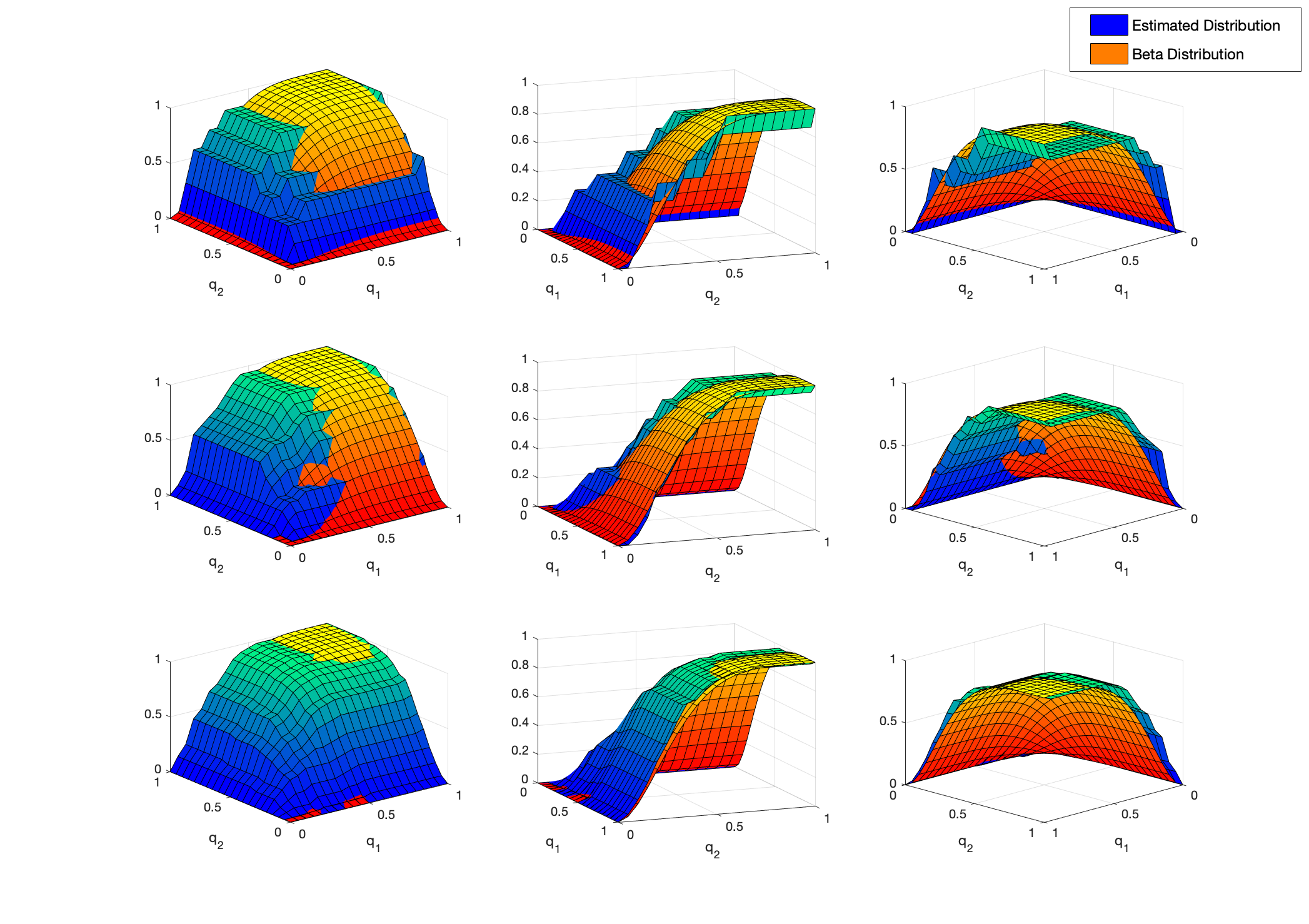}
\caption{Three different views of the estimated distribution and the ``true" joint $Beta(2,5)$ distribution, for different numbers of drinking episodes $m=7$, $m=16$, and $m=42$ in the top, middle, and bottom rows, respectively, for a fixed number of nodes $M=400$ and level of discretization $N=128$.}
\label{fig:3DSim}
\end{figure}

Next, we show that as the number of nodes $M$ and the level of discretization $N$ increases, the normalized sum of squared differences at each node between the estimated and the ``true" distribution decreases. First, we fixed the level of discretization at $N = 128$, and we increased the number of nodes $M$. Let
\begin{align*}
    \bar{D}_M = \frac{1}{M} \sum_{j=1}^{M} (p_j - b_j)^2.
\end{align*}
In Table (\ref{tab:2}), we observe that for the fixed value of $N = 128$, as the number of nodes, $M$, increases, the normalized sum of the squared differences at each node between the estimated distribution and the ``true" distribution, $\bar{D}_M$, decreases.
\begin{table}[H]
\centering
    \caption{Decrease in $\bar{D}_M$, the normalized sum of the squared differences at each node between the estimated distribution and the ``true" distribution, with increasing $M$, the number of nodes, for a fixed value for the level of discretization, $N=128$.}
    \label{tab:2}
    \begin{tabular}{ll}
    \hline\noalign{\smallskip}
    {$M$} & {$\bar{D}_M$} \\
    \noalign{\smallskip}\hline\noalign{\smallskip}
    25 & 0.03795 \\
    100 & 0.03025 \\
    225 & 0.02974 \\
    400 & 0.02081 \\
    \noalign{\smallskip}\hline
    \end{tabular}
\end{table}

Next, we fixed the number of nodes at $M = 400$, and we increased the level of discretization $N$. Let
\begin{align*}
    \bar{D}_N = \frac{1}{N} \sum_{j=1}^{M} (p_j - b_j)^2.
\end{align*}

In Table (\ref{tab:3}), we observe that for $N = 128$ fixed, as the number of nodes, $M$, increases, the normalized sum of the squared differences at each node between the estimated distribution and the ``true" distribution, $\bar{D}_N$, decreases.
\begin{table}[H]
\centering
    \caption{Decrease in $\bar{D}_N$, the normalized sum of the squared differences at each node between the estimated distribution and the ``true" distribution, with increasing $N$, the level of discretization, with the number of nodes fixed at $M=400$.}
    \label{tab:3}
    \begin{tabular}{ll}
    \hline\noalign{\smallskip}
    {$N$} & {$\bar{D}_N$} \\
    \noalign{\smallskip}\hline\noalign{\smallskip}
    4 & 1.22783 \\
    16 & 0.65644 \\
    64 & 0.18565 \\
    128 & 0.06504 \\
    \noalign{\smallskip}\hline
    \end{tabular}
\end{table}

The choice of the ``true" distribution for $q_1$ and $q_2$ in the simulation case is the scatterplot of samples for a set of 18 drinking episodes including BrAC and TAC measurements of different individuals obtained by a deterministic approach in \cite{Banks:2018}. However, the $Beta(2,5)$ distribution was chosen strictly for the purpose of demonstration. When applying our algorithm to actual clinic or lab collected Human subject data, an advantage of our nonparametric approach is that we do not need to make any assumptions about the family of feasible distributions for the parameter vector. In addition, the independent and identically distributed assumption was also very simplistic given that $q_1$ and $q_2$ parameters depend on the same individual and environmental conditions at the time of measurements. This assumption is also relaxed in the flexible nonparametric approach used in the next example. 

\subsection{Example 2: Estimation Based on Actual Human Subject Data}
\label{sec:7.2}

The two datasets used in this example were obtained by two different alcohol biosensors; SCRAM CAM$\textsuperscript{\textregistered}$ and WrisTAS$^{\text{TM}}$7 \cite{Luczak:2015,Saldich:2020}. We fixed the number of nodes at $M=400$ and the level of discretization at $N=128$, both sufficiently large with respect to convergence as we observed in our simulation data examples. From each dataset, we chose $m = 9$ different drinking episodes. We split the drinking episodes into a training set consisting of $8$ drinking episodes, and a testing set consisting of $1$ drinking episode. This way, we could apply the leave-one-out cross-validation (LOOCV) method. We repeated this partitioning process $9$ times, each time leaving out a different drinking episode. Using the training set, we first estimated the distribution of the parameter vector $\bm{q} = (q_1,q_2)$. Next, we sampled $100$ parameter vectors $\bm{q}=(q_1,q_2)$ from the estimated distribution, and using those along with the BrAC input from the testing dataset, we simulated $100$ TAC longitudinal signals. From these $100$ simulated TAC signals, we estimated the ``true" TAC by computing the mean at each time, and we provided what we refer to as a $95\%$ conservative error band, or simply as a $95\%$ error band, by taking the 2.5 and 97.5 percentiles. This approach for the error band is also used for a number of statistics associated with the TAC curve that are of particular interest to researchers and clinicians working in the area of alcohol use disorder. 

For the first example, we considered the dataset collected using the SCRAM alcohol biosensor. Prior to applying the leave-one-out cross-validation (LOOCV) method, in order to visualize the estimated density and distribution of $\bm{q} = (q_1, q_2)$, and the marginal densities of $q_1$ and $q_2$, we trained the algorithm on all 9 drinking episodes. Figure (\ref{fig:4plots}) illustrates the four aforementioned plots. In the estimated density plot, we can see that our numerical result for this example is in agreement with the theoretical result in Lindsay and Mallet \cite{Lindsay:1983,Mallet:1986}, which states that the maximum likelihood estimator $\hat{P}_{\bm{n},m}$ can be found in the class of discrete distributions with at most $m$ support points, i.e. $\hat{P}_{\bm{n},m} \in \mathcal{P}_M(Q)$, where $M \leq m$. In this example, since we had 9 drinking episodes, the estimated density plot displays the support points among 400 nodes for $\bm{q} = (q_1, q_2)$.


In addition, for this sample, the sample mean of $\bm{q} = (q_1, q_2)$ is calculated to be $$\bar{\bm{q}} = (0.6003, 1.2452),$$ the sample covariance matrix is calculated to be
$$\bm{S}_{\bm{q}} = \begin{pmatrix}
    0.0706 & -0.0264 \\
    -0.0264 & 0.0483
\end{pmatrix},$$
and the sample correlation is calculated to be $\rho_{\bm{q}} = -0.4519$. Based on this, we observe that for our training population consisting of 9 drinking episodes there is a moderate negative association between the parameters $q_1$ and $q_2$.

\begin{figure}[H]
\centering
\includegraphics[width=12.2cm]{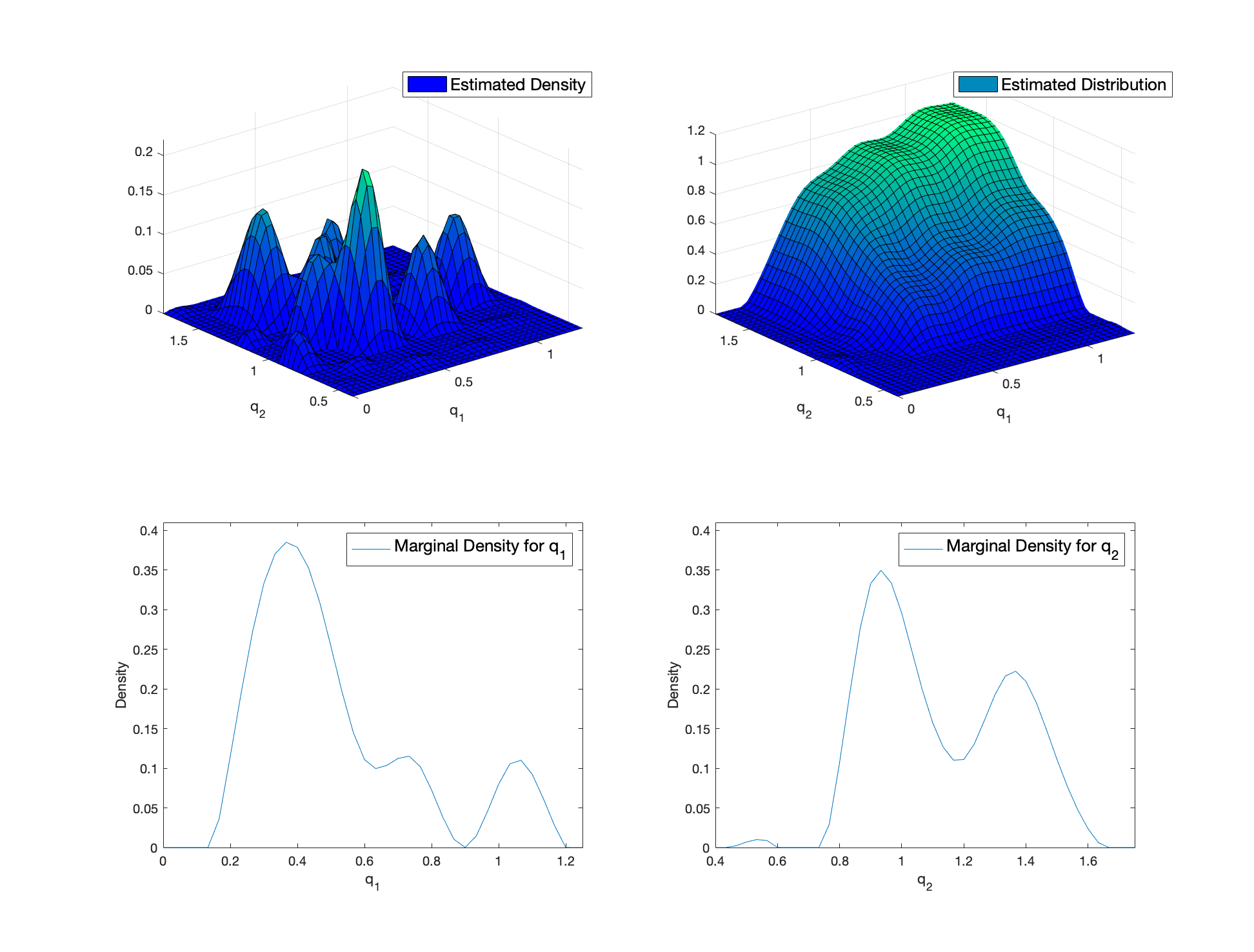}
\caption{The estimated density (top left), estimated distribution (top right), marginal density of $q_1$ (bottom left), and marginal density of $q_2$ (bottom right) obtained from $m=9$ drinking episodes collected using the SCRAM Alcohol Biosensor, for a fixed number of nodes $M=400$ and level of discretization $N=128$.}
\label{fig:4plots}
\end{figure}
We applied the leave-one-out cross-validation (LOOCV) method as explained above to the 9 drinking episodes from the SCRAM biosensor. Figure (\ref{fig:subplots_SCRAM}) shows the measured TAC (i.e. measured by the SCRAM alcohol biosensor) and the estimated TAC (i.e. obtained from our algorithm) for all the 9 drinking episodes left out in the testing set in the partitioning process, and the conservative $95 \%$ error band for a fixed number of nodes $M=400$ and level of discretization $N=128$.

Alcohol researchers and clinicians are particularly interested in a certain statistics associated with drinking episodes: the maximum or peak value of the TAC curve, the time at which the peak value of the TAC is attained, and the area under the TAC curve. The area under the curve (AUC) is a quantifying measure of exposure to the alcohol that integrates the transdermal alcohol concentration across time. Tables (\ref{tab:4})-(\ref{tab:6}) display these statistics along with the measured (or actual) value obtained by the SCRAM alcohol biosensor as well as the conservative $95\%$ error band for the 9 drinking episodes from the testing set. From these tables, we can observe that the $95\%$ error bands do a reasonably good job of capturing the actual values of these statistics specially for the value of the peak TAC and the area under the curve displayed in Table (\ref{tab:4}) and (\ref{tab:6}). In Figure (\ref{fig:subplots_SCRAM}), we can see that there are minor fluctuations in the measured TAC curve. If we smoothed the measured TAC curve, we would have obtained better results for capturing the time at which the peak value of the TAC is attained displayed in Table (\ref{tab:5}).

\begin{figure}[ht]
\centering
\includegraphics[height= 8.1cm]{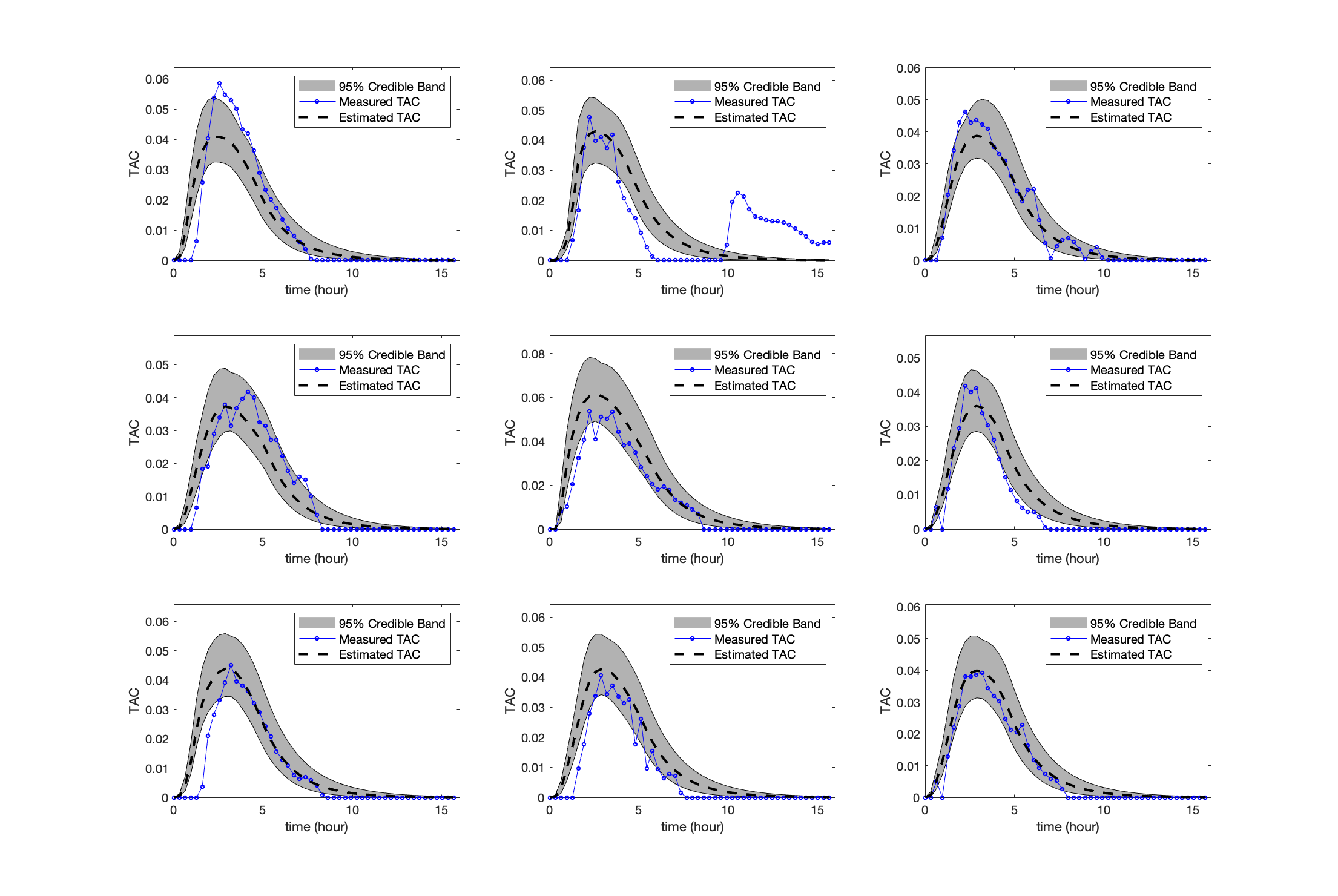}
\caption{The measured TAC, the estimated TAC, and the conservative $95 \%$ error band for 9 drinking episodes from the testing set collected using the SCRAM alcohol biosensor using the LOOCV.}
\label{fig:subplots_SCRAM}
\end{figure}

\begin{table}[ht]
\centering
    \caption{The measured peak TAC, estimated peak TAC, and the $95 \%$ error band for the 9 drinking episodes from the testing set collected using the SCRAM alcohol biosensor.}
    \label{tab:4}
\begin{tabular}{clll}
  \hline\noalign{\smallskip}
Drinking Episode & Measured Peak TAC & Estimated Peak TAC & 95\% Error Band \\ 
  \noalign{\smallskip}\hline\noalign{\smallskip}
  1 & 0.0585 & 0.0413 & (0.0327, 0.0539) \\ 
  2 & 0.0477 & 0.0433 & (0.0324, 0.0543) \\ 
  3 & 0.0464 & 0.0391 & (0.0320, 0.0501) \\ 
  4 & 0.0417 & 0.0375 & (0.0301, 0.0489) \\ 
  5 & 0.0535 & 0.0618 & (0.0497, 0.0781) \\ 
  6 & 0.0419 & 0.0361 & (0.0287, 0.0465) \\ 
  7 & 0.0450 & 0.0441 & (0.0346, 0.0558) \\ 
  8 & 0.0405 & 0.0430 & (0.0345, 0.0542) \\ 
  9 & 0.0391 & 0.0402 & (0.0313, 0.0508) \\ 
  \noalign{\smallskip}\hline
\end{tabular}
\end{table}

\begin{table}[H]
\centering
    \caption{The measured peak time, estimated peak time, and the $95 \%$ error band for the 9 drinking episodes from the testing set collected using the SCRAM alcohol biosensor.}
    \label{tab:5}
\begin{tabular}{clll}
  \hline\noalign{\smallskip}
Drinking Episode & Measured Peak Time & Estimated Peak Time & 95\% Error Band \\ 
  \noalign{\smallskip}\hline\noalign{\smallskip}
  1 & 2.5600 & 2.4480 & (1.9200, 2.8800) \\ 
  2 & 2.2400 & 2.6080 & (2.2400, 2.8800) \\ 
  3 & 2.2400 & 2.8928 & (2.5600, 3.2000) \\ 
  4 & 4.1600 & 2.9056 & (2.5600, 3.2000) \\ 
  5 & 2.2400 & 2.5728 & (2.2400, 2.8800) \\ 
  6 & 2.2400 & 2.8928 & (2.5600, 3.2000) \\ 
  7 & 3.2000 & 2.9696 & (2.5600, 3.2000) \\ 
  8 & 2.8800 & 2.9312 & (2.5600, 3.2000) \\ 
  9 & 3.2000 & 2.9088 & (2.5600, 3.2000) \\ 
  \noalign{\smallskip}\hline
\end{tabular}
\end{table}

\begin{table}[H]
\centering
    \caption{The measured area under the curve (AUC), estimated AUC, and $95 \%$ error band for 9 drinking episodes from the testing set collected using the SCRAM alcohol biosensor.}
    \label{tab:6}
\begin{tabular}{clll}
  \hline\noalign{\smallskip}
Drinking Episode & Measured AUC & Estimated AUC & 95\% Error Band \\ 
  \noalign{\smallskip}\hline\noalign{\smallskip}
  1 & 0.1909 & 0.1784 & (0.1382, 0.2229) \\ 
  2 & 0.1876 & 0.1799 & (0.1343, 0.2321) \\ 
  3 & 0.1868 & 0.1751 & (0.1364, 0.2421) \\ 
  4 & 0.1765 & 0.1735 & (0.1349, 0.2394) \\ 
  5 & 0.2231 & 0.2963 & (0.2203, 0.3911) \\ 
  6 & 0.1151 & 0.1496 & (0.1117, 0.1982) \\ 
  7 & 0.1474 & 0.1912 & (0.1444, 0.2563) \\ 
  8 & 0.1277 & 0.1898 & (0.1412, 0.2505) \\ 
  9 & 0.1493 & 0.1750 & (0.1307, 0.2319) \\ 
  \noalign{\smallskip}\hline
\end{tabular}
\end{table}

For the second example, we applied the leave-one-out cross-validation (LOOCV) method, as explained before, to the 9 drinking episodes from the WrisTAS7 alcohol biosensor. Figure (\ref{fig:subplots_WrisTAS}) shows the measured TAC (i.e. measured by the WrisTAS7 alcohol biosensor) and the estimated TAC (i.e. obtained from our algorithm) for all the 9 drinking episodes left out in the testing set in the partitioning process, and the conservative $95 \%$ error band for a fixed number of nodes $M=400$ and level of discretization $N=128$. We can observe that we obtained similar results as the previous example.
\begin{figure}[H]
\centering
\includegraphics[height= 8.1cm]{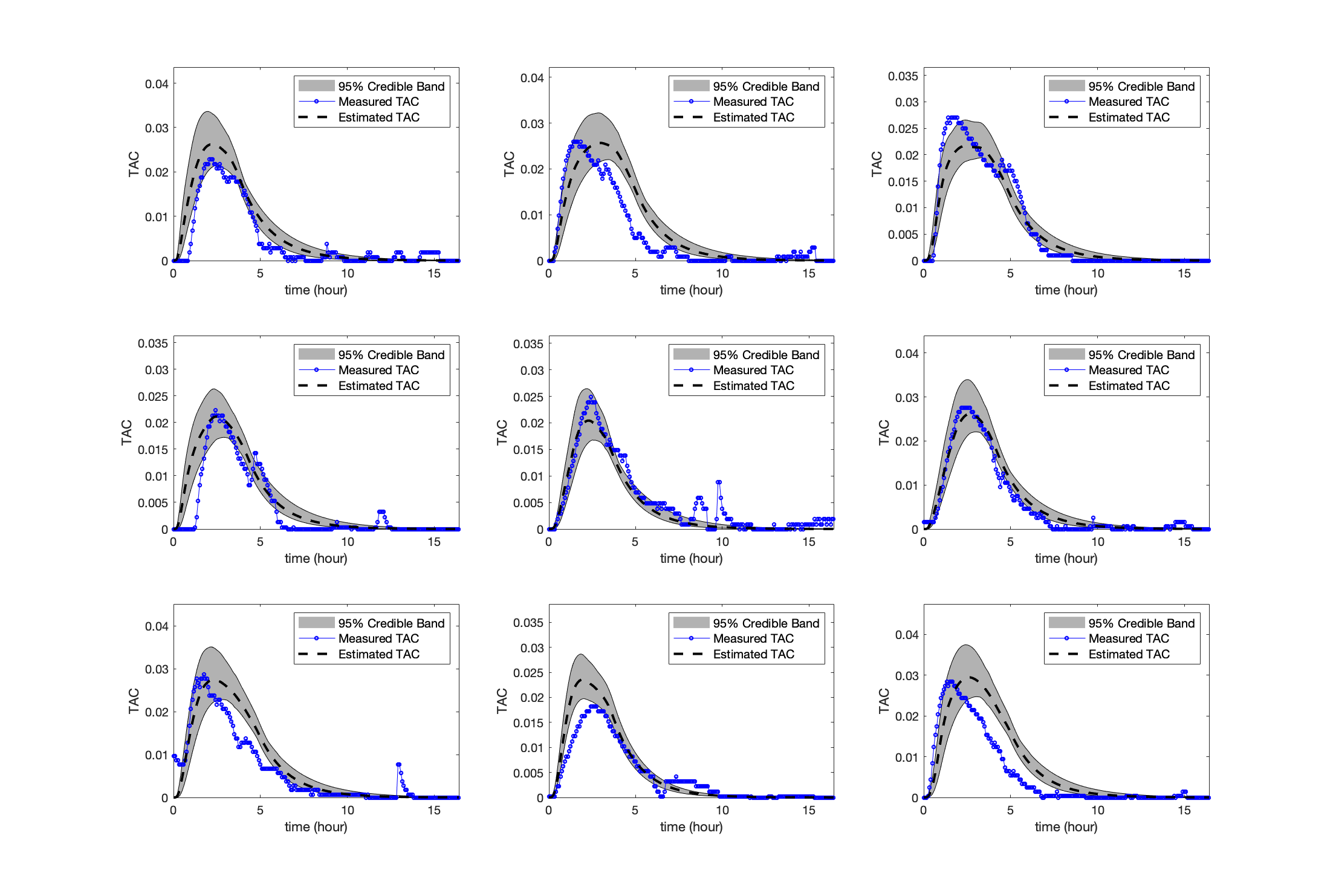}
\caption{The measured TAC, the estimated TAC, and the conservative $95 \%$ error band for 9 drinking episodes from the testing set collected using the WrisTAS7 alcohol biosensor using the LOOCV.}
\label{fig:subplots_WrisTAS}
\end{figure}


\section{Discussion and Concluding Remarks}
\label{sec:8}
In this paper, we considered the nonparametric fitting of a population model for the transdermal transport of alcohol based on a maximum likelihood approach applied to a mixed effects statistical model. In estimating a population model, we were actually estimating the distribution of the model parameters and consequently the MLE problem was formulated as an optimization problem over a space of feasible probability measures endowed with the weak topology induced by the Prohorov metric. 

By using a first principles physics based model in the form of a one dimensional diffusion equation, we were able to capture the essential features of transdermal transport while keeping the dimension of the parameter space low. In this way, we were able to avoid having to introduce regularization so as to mitigate ill-posedness and over-fitting. On the other hand, the fact that the model was infinite dimensional being based on a partial differential equation, computing the MLE necessitated finite dimensional approximation. 

We were able to first theoretically demonstrate the existence and then the consistency of our MLE using a decades old result from the literature. The consistency result is with respect to the uncertainty across subjects. It is likely that the consistency results proved in \cite{Banks:2018} and \cite{Banks:2012}, in the context of a naive pooled statistical model based on a nonlinear least squares estimator, for problems either the same as, or very similar to the one we consider here, would apply for the uncertainty within each subject (i.e. as the resolution of the data with respect to time increases). At present, this is just a hypothesis and a possible avenue for future research; as of yet, we have not carefully examined this possibility. 

In addition, we were able to use linear semigroup theory, in particular the Trotter-Kato Theorem, and the properties of the weak topology and the Prohorov metric on the space of feasible probability measures, to establish a convergence result with respect to the MLE for the finite dimensional approximating estimation problems and the MLE for the estimation problem posed in terms of the original underlying infinite dimensional model. 

We were able to demonstrate the efficacy of our theoretical results numerically first on an example involving simulated data and then on one involving actual human subject data from an NIH funded study. We used our scheme to obtain the joint density and distribution of the parameters as well as estimates and conservative $95\%$ error bands for the TAC signal and a number of TAC related statistics of particular interest to researchers and clinicians who work in the area of alcohol use disorder.

In addition to consistency with respect to the intrinsic uncertainty, other extensions we are currently looking at include the development of a general framework for estimating random parameters in general finite or infinite dimensional, continuous or discrete-time dynamical systems (e.g. ODEs, PDEs, FDEs, DEs, etc.) that would potentially subsume the results presented here as well as in \cite{Banks:2018} and \cite{Banks:2012}. 

Finally, since the actual motivation for this investigation is the development of schemes for converting biosensor measured TAC into BAC/BrAC, the next step would be to examine how well population models, estimated using the approach we have presented here, perform when used as part of a scheme that deconvolves an estimate for BAC/BrAC from the TAC signal. In particular, we are interested in comparing it to the schemes, used for this same purpose, developed and implemented in \cite{Hawekotte:2021} and \cite{Sirlanci:2018}. In addition, we are also interested in examining how our uncertainty quantification scheme for the TAC to BAC/BrAC conversion problem performs when compared to the non-physics based, machine learning inspired schemes developed in \cite{Fairbairn:2021}, \cite{Oszkinat:2021} and \cite{Oszkinat:2021A}.

\begin{acknowledgements}
We thank the Luczak laboratory students and staff members, particularly Emily Saldich, for their assistance with data collection and management for the SCRAM biosensor. We also thank Dr. Tamara Wall for providing the data for the WrisTAS7 biosensor.
\end{acknowledgements}
\section{Declarations}
\textbf{Funding:} This study was funded in part by the National Institute on Alcohol Abuse and Alcoholism (Grant Numbers: R21AA017711 and R01AA026368, S.E.L. and I.G.R.) and by support from the USC Women in Science and Engineering (WiSE) program (L.A.).\\
\textbf{Conflict of Interest:} The authors declare that they have no conflicts of interest.\\
\textbf{Availability of Data and Material:} The data used in this study can be made available upon special request to the authors.\\
\textbf{Code Availability:} The codes used in this study can be made available upon special request to the authors.



\end{document}